\documentclass[letterpaper, DIV=15]{scrartcl}  	
\usepackage{amssymb, amsthm, mathtools,bbm,color}
\usepackage[square,sort,comma,numbers]{natbib}
\usepackage{enumerate}
\usepackage{dsfont,times,tikz} 
\usetikzlibrary{matrix,arrows.meta,shapes.geometric,calc}
\definecolor{Plum}{rgb}{.5,0,1}

\newcommand{\ran}{\operatorname{ran}}

\newcommand{\gap}{\operatorname{gap}}
\newcommand{\spec}{\operatorname{spec}}
\newcommand{\supp}{\operatorname{supp}}

\newcommand{\spa}{\operatorname{span}}

 \def\1{{\mathchoice {\mathrm{1\mskip-4mu l}} {\mathrm{1\mskip-4mu l}} %
		{\mathrm{1\mskip-4.5mu l}} {\mathrm{1\mskip-5mu l}}}}

\newcommand{\bC}{{\mathbb C}}
\newcommand{\bN}{{\mathbb N}}
\newcommand{\bZ}{{\mathbb Z}}
\newcommand{\cA}{{\mathcal A}}
\newcommand{\cH}{{\mathcal H}}
\newcommand{\braket}[2]{\langle {#1} \ | \ {#2}\rangle}
\newcommand{\ket}[1]{|{#1}\rangle}
\newcommand{\ketbra}[1]{\vert #1\rangle\langle #1\vert}
\newcommand{\vp}{\varphi}
\newcommand{\caG}{\mathcal{G}}

\newcommand{\caC}{\mathcal{C}}

\newcommand{\cG}{\mathcal{G}}
\newcommand{\cB}{\mathcal{B}}

\newcommand{\cD}{\mathcal{D}}

\newcommand{\cS}{\mathcal{S}}
\newcommand{\cR}{\mathcal{R}}
\newcommand{\cT}{\mathcal{T}}

\newcommand{\cV}{\mathcal{V}}

\newcommand{\per}{\mathrm{per}}
\newcommand{\obc}{\mathrm{obc}}

\renewcommand{\subset}{\subseteq}

\newcommand{\be}{\begin{equation}}
\newcommand{\ee}{\end{equation}}

\newtheorem{theorem}{Theorem}
\newtheorem{lem}[theorem]{Lemma}
\newtheorem{cor}[theorem]{Corollary}

\numberwithin{equation}{section}
\numberwithin{theorem}{section}
\usepackage{graphicx}				
\usepackage{amssymb}
\title{\LARGE The spectral gap of a fractional quantum Hall system\\ on a thin torus}
\author{Simone Warzel and Amanda Young}
\date{\small\today}							

\begin{document}
\maketitle

\minisec{Abstract} 
We study a fractional quantum Hall system with maximal filling $ \nu = 1/3 $ in the thin torus limit. The corresponding Hamiltonian is a truncated version of Haldane's pseudopotential, which upon a Jordan-Wigner transformation is equivalent to a one-dimensional quantum spin chain with periodic boundary conditions. Our main result is a lower bound on the spectral gap of this Hamiltonian, which is uniform in the system size and total particle number. The gap is also uniform with respect to small values of the coupling constant in the model. The proof adapts the strategy of individually estimating the gap in invariant subspaces used for the bosonic $ \nu = 1/2 $ model to the present fermionic case.

\section{Introduction}\label{sec:intro}
\subsection{Hamiltonian description of FQHE on the torus}
Haldane pseudopotentials provide a Hamiltonian description of the main features for factional quantum Hall (FQH) systems~\cite{PhysRevLett.51.605,Duncan1990}. 
These are parent Hamiltonians for Laughlin's famous many-particle wave functions corresponding to a fixed filling ratio $ \nu $ in a two-dimensional geometry with a perpendicular constant magnetic field \cite{Trugman:1985lv,MooreRead91,Ortiz:13,Lee:13}, and have been shown to emerge in a scaling limit of short-range interactions \cite{seiringer:2020}. 

An example of a non-trivial  two-dimensional geometry is a rectangular torus of lengths $ a ,b > 0 $. In the presence of the magnetic field all lengths are naturally measured in terms of the magnetic length $ \ell > 0 $. For  integer flux $ L := \frac{a b }{2\pi \ell^2 } \in \mathbb{N} $, the lowest Landau level (LLL) on the torus is an $L$-dimensional Hilbert space, which is spanned by the orthonormal functions
\be\label{torus_states}
\psi_{[m]}(x,y) := \sum_{j\in\bZ} \psi_{m+j L}(x,y), \qquad 0\leq x < a, \; 0\leq y < b,
\ee
for $1\leq m \leq L$ where  $\alpha := \frac{2\pi\ell}{b}$ and 
\[
\psi_{m}(x,y) := \left(\frac{\alpha}{2\pi^{3/2}}\right)^{1/2} e^{i\frac{\alpha}{\ell} my}e^{-\frac{1}{2}(x/\ell-\alpha m)^2} , \quad m\in\bZ , 
\]
constitutes an orthogonal basis of Landau orbital's on the infinite cylinder. The one-particle basis~\eqref{torus_states} of the LLL on the torus can thus be identified with the $ L $  sites of the integer lattice $ \Lambda_L =[1,L]$ in the ring geometry.  In contrast to the planar, spherical or cylinder geometry, various suggestions (deviating only in small details) exist in the physics literature for the precise form of Haldane's pseudopotentials on the torus \cite{Ortiz:13,Lee:13,Lee:15,Haldane:2018pi}. In this geometry, the most basic form of Laughlin's $ N $-particle wave function with filling fraction $ \nu= (p+2)^{-1} $  with $ p \in \mathbb{N}_0 $ is proportional to \cite{PhysRevLett.50.1395,HR85,Frem16} 
\[
 \exp\left( - \sum_{j} \frac{x_j^2}{2\ell^2}\right)  \ \prod_{j< k } \vartheta_1\left( \frac{z_j - z_k}{a} , i \frac{a}{b} \right)^{p+2} 
\]
when written in complex coordinates $ z_j = x_j + i y_j $. The Jastrow factor involves
Jacobi's theta function $ \vartheta_1(z,\tau) $, which has its zeros at $ z \in \bZ + \tau \bZ $. 
 Following the pragmatic approach of taking the second-quantization of Trugman-Kivelson's singular pair-interactions  \cite{Trugman:1985lv} projected onto the LLL, which has wave functions with the above Jastrow-factor in its kernel, one arrives at the Hamiltonian
\be\label{HaldaneH}
 W_{\Lambda_L}^\per  = \sum_{s  }  B_s^* B_s ,  \quad \mbox{with} \quad B_s := {\sum_k  }' F_p(k)  \, c_{s+k} c_{s-k}  ,
\ee
which is expressed in terms of creation and annihilation operators $ c_m^*, c_m  $ of the single-particle orbitals~\eqref{torus_states}. 
Here the first summation is over all integers and half-integers with  $ \frac{1}{2} \leq s \leq L  $. The second summation is over $ k  $ with $ \frac{1}{2} \leq k \leq N $ such that $  s + k $ (and hence $  s-k $) is integer. Moreover, as will always be the case for the periodic system, additions are understood modulo $ L $. 
The function $ F_p $ depends on the filling fraction $ \nu= (p+2)^{-1} $   under consideration. For the case of interest in this paper,  $ \nu = 1/3 $, one has 
\[ F_{1}(k) = \sum_{j\in\bZ}\alpha(k+jN)e^{-\alpha^2(k+jN)^2} . \] 
The $ \nu = 1/2 $-pseudopotential, which stems from a delta-pair interaction on the LLL,  is represented by $ F_{0}(k) = \sum_{j\in\bZ} e^{-\alpha^2(k+jN)^2} $. Above, the (anti-)commutation relations of the creation and annihilation operators must be chosen appropriately to reflect that odd $ p $ correspond to fermionic models and even $ p $ to bosonic models. 

In the thin torus limit $ \alpha \to \infty $, it is suggestive to truncate the $ k $-sum in~\eqref{HaldaneH} to the first few leading terms.  For $ \nu = 1/3$, a  truncation at $ | 2k| \leq 1 $ would yield the so-call Tao-Thouless Hamiltonian~\cite{tao:1983}, which consists of the commuting local terms 
$  n_m n_{m+2}  $ and $ n_{m+1} n_{m+2} $. Such a Hamiltonian trivially has a spectral gap, but does not describe a non-zero Hall conductivity~\cite{kapustin:2020} (see also~\cite{Bachmann:2021dp} and references therein). 
The first non-trivial truncation is at $ |2k | \leq 3 $, which, upon neglecting another overall multiplicative constant, corresponds to approximating  $  W_{\Lambda_L}^\per $ by 
\[
\sum_{m \in \Lambda_L} \left(  n_m n_{m+2} +  |F_1(1)|^{-2} b_m^* b_m \right) , \quad n_m := c_m^* c_m , \quad  b_m := F_1(\tfrac{1}{2}) \, c_{m+1} c_{m+2} + F_1(\tfrac{3}{2}) \,  c_{m} c_{m+3} .
\]
As the original operator $  W_{\Lambda_L}^\per  $, its truncated version conserves the total particle number as well as center of mass. 
Truncated Haldane pseudopotentials of the above form have been studied in~\cite{PhysRevB.85.155116,Jansen:2012da,Nakamura:2012bu,wang:2015,NWY:2021} and are believed to capture, at least qualitatively, the main features of FQH systems such as their incompressibility (see also~\cite{Frem15,Frem16}). In the Hamiltonian description, the latter is explained by combining the maximal filling factor $ \nu $ of their degenerate ground-state space with a spectral gap above these ground states, which is uniform in the particle number as well as the volume $ \Lambda_L $. For the case $ \nu= 1/3 $ these properties were mathematically established in~\cite{NWY:2021}. However, the estimates of the bulk gap were plagued by edge modes. It is the purpose of the present paper to revisit and fix this problem.

\subsection{Bulk spectral gap of the corresponding spin chain} 
As in~\cite{NWY:2021} we find it convenient to rewrite the fermionic system as a spin-$ \tfrac{1}{2} $ chain 
using the Jordan-Wigner transformation on  $ \Lambda_L  = [1,L] \cap \bZ $. Given the three Pauli matrices $ \sigma_x^1 , \sigma_x^2 , \sigma_x^3 $ and the corresponding lowering and raising operators 
$ \sigma_x^\pm := \tfrac{1}{2} ( \sigma_x^1 + i \sigma_x^2 ) $, 
which act on $ \mathbb{C}^2 \equiv \spa\{ | 1 \rangle , | 0 \rangle \} $  for each $ x $, the operators 
\be\label{eq:JW}
c_x =  P_{x-1} 	\sigma_x^-  , \qquad 
c_x^*  =	 P_{x-1}  \sigma_x^+ , \qquad P_{x-1} := \prod_{k=1}^{x-1} \sigma_k^3 .
\ee
implement the canonical anticommutation relations \cite{Jordan:1928sd} on the tensor product 
\begin{equation*}
\mathcal{H}_{\Lambda_L} := \bigotimes_{x=1}^L \mathbb{C}^2.
\end{equation*} 
As usual, we suppress the identity $ \1 $ from our notation when considering the Pauli matrices on $\cH_{\Lambda_L}$. Our conventions are such that $\sigma^3\ket{k} = (-1)^{k+1}\ket{k}$ for $k=0,1$, and $P_0:=\1$, from which one finds that the parity operator satisfies $ P_L = \exp\left( i \pi \sum_{x =1}^L ( n_x -1) \right) = (-1)^{N+L}  $, where $ N =\sum_{x=1}^Ln_x$ is the total number of particles operator. 

Under the Jordan-Wigner transformation, the truncated Haldane $ \nu = 1/3 $-pseudopotential with periodic boundary conditions is unitarily equivalent to the following spin-$\tfrac{1}{2} $ chain:
\be\label{def:Hspin}
H_{\Lambda_L}^\per  :=  \sum_{x=1}^L \left( n_{x}n_{x+2} + \kappa q_x^*q_x \right) 
\ee
where $ \kappa  > 0 $ and $ \lambda \in \bC $ will subsequently be taken as arbitrary parameters. The `physical' values correspond to $ \kappa = | F_1(\tfrac{1}{2})|^2/ |F_1(1)|^{2} $ and $ \lambda =- F_1(\tfrac{3}{2}) / F_1(1) $. For all $ x \in \Lambda_L\backslash \{ L-2,L\}$, under the  transformation~\eqref{eq:JW} the nontrivial four-spin coupling $ q_x^*q_x $ is consistent with the definition
\be\label{def:q}
 q_x = \sigma_{x+1}^- \sigma_{x+2}^- - \lambda \  \sigma_{x}^- \sigma_{x+3}^- \, .
\ee
For $ x \in \{ L-2,L\}  $, one obtains $ q_x = \sigma_{x+1}^- \sigma_{x+2}^- + \lambda P_L \  \sigma_{x}^- \sigma_{x+3}^- $, which is equivalent to~\eqref{def:q} on all subspaces of fixed total particle number $ N $ with $ L + N $ odd. By redefining the Jordan-Wigner transformation replacing $ 	\sigma_x^- \mapsto e^{i\pi x/L} 	\sigma_x^- $, one finds that on all subspaces of fixed total particle number $ N $ with $ L + N $ even, the transformed Hamiltonian is isospectral to~\eqref{def:Hspin} with 
$ q_x $ defined as in~\eqref{def:q} for all  $ x \in \Lambda_L $. Since we are merely interested in spectral information, we will subsequently analyze the spectrum of~\eqref{def:Hspin} with $ q_x $ as in~\eqref{def:q} for all $ x \in \Lambda_L $.\\

There are two main strategies for proving spectral gaps in quantum spin Hamiltonians: (1) inductive methods based on norm-estimates of nested ground-state projections, such as the martingale method \cite{affleck:1988,nachtergaele:1996}, and (2) finite-size criteria, such as the approach by Knabe \cite{knabe:1988}. Both methods can be applied to open and periodic boundary conditions \cite{young:2016,lemm:2018,Lemm:2020}. To obtain results for the periodic system,  in practice one often uses the gap bound from the martingale method for the system with open boundary conditions for the finite-size criteria. However, if the model with open boundary conditions has low-lying edge excitations that are not present in the periodic system, this strategy is at best sub-optimal as it produces gap estimates that do not accurately reflect the bulk behavior. 
This was the case of the spectral-gap results in \cite{NWY:2021}, where all estimates on spectral gaps depend on lower bounds of the spectral gap for the Hamiltonian with the open boundary conditions, i.e.,
\begin{equation}\label{def:Hopen}
	H_{\Lambda_L} =\sum_{x=1}^{L-2}n_{x}n_{x+2} + \kappa\sum_{x=1}^{L-3}q_x^*q_x, \end{equation}
with $ q_x $ as in~\eqref{def:q}. 
As was explained in more detail in~\cite{NWY:2021}, this Hamiltonian has a multitude of eigenstates associated with excitations at the left or right boundary of the interval $ \Lambda_L =[1,L] \subset \bZ $, whose eigenvalues tend to zero as $ \lambda \to 0 $. For example, if $ L \geq 6 $ the subspace $ \spa\{ | 11001 0\dots\rangle , | 10110 0\dots\rangle  \} $, which is spanned by tensor products of the canonical basis vectors of $ \bC^2 $,  is invariant under $ H_{\Lambda_L} $. For small $ | \lambda | $, the smaller of the two eigenvalues on this subspace is of the order $ \frac{\kappa}{1+\kappa} |\lambda|^2 + \mathcal{O}(|\lambda|^4)  $. However, as suggested by the numerics in \cite{NWY:2021}, such modes do not appear in the system with periodic boundary conditions and, as such, do not accurately reflect the bulk gap.

 In \cite{WY:2021} we introduced a strategy to circumvent such edge modes for the Hamiltonian with open boundary conditions and produce a robust estimate on the bulk gap in the case of a truncated version of Haldane's $\nu=1/2$ pseudopotential. The aim of this work is to adapt this strategy to the $\nu=1/3$ model. Our main result is the following theorem.

\begin{theorem}\label{thm:bulk_gap}
 There is a monotone increasing function $ f : [0,\infty) \to [0,\infty) $ such that for all $ 0\neq \lambda \in \mathbb{C} $  with the property $ f(|\lambda|^2)   < 1/3 $ and all $ \kappa \geq 0 $:
\be\label{eq:main2}
\liminf_{L\to \infty} \gap(H_{\Lambda_L}^\per)   \geq  \min\left\{ \gamma^\per , \, \frac{\kappa}{6(1+2|\lambda|^2)} \left( 1 - \sqrt{3 f(|\lambda|^2) } \right)^2 \right\} 
\ee
where
\be\label{def:gammaper}
	 \gamma^\per:=\frac{1}{3}\min\left\{1,\frac{\kappa}{2+2\kappa|\lambda|^2},\frac{\kappa}{1+\kappa}\right\}.
\ee
\end{theorem}
The function $f(r)$ (defined in Theorem~\ref{thm:MM} below) was analyzed in \cite[Appendix A]{NWY:2021} where it was shown that $f(|\lambda|^2)<1/3$ for $|\lambda|<5.3$. 
The bound~\eqref{eq:main2} is a significant improvement over \cite[Theorem 1.2]{NWY:2021} since it remains open in the  limit $ \lambda \to 0 $. It is consistent with the numerical results in~\cite{Nakamura:2012bu,wang:2015,NWY:2021} and the variational calculation in~\cite{NWY:2020}. The latter is based on a quasi-particle-hole ansatz (`magnetoroton' \cite{yang:2012}) and predicts the asymptotic parabolic behavior $ 1- \frac{2\kappa}{\kappa -1} |\lambda|^2 + \mathcal{O}(|\lambda|^4) $ as $ \lambda \to 0 $ for the spectral gap when  $ \kappa > 1 $ .  While the RHS of~\eqref{eq:main2} is not sharp, it is consistent with this asymptotics. \\

The paper is organized as follows. In the next section, a slew of invariant subspaces for both the Hamiltonian of interest, $H_{\Lambda_L}^\per $, and its counterpart with open boundary conditions, $ H_{\Lambda_L} $ are classified. Key to the invariant subspace gap strategy in~\cite{WY:2021} is identifying invariant subspaces that contain ground states. We focus on such subspaces and discuss a number of additional relevant properties for proving Theorem~\ref{thm:bulk_gap}. The invariant subspace strategy for the gap is then reviewed and applied to the present model in Section~\ref{sec:proof}. The arguments in this paper not only improve the estimate on the bulk gap in~\cite{NWY:2021}, but also streamline parts of the arguments presented there.

\section{Invariant subspaces and ground states}

One way to identify an invariant subspace of a Hamiltonian is to construct a subspace that is invariant under each of its interaction terms. For the Hamiltonians from~\eqref{def:Hspin} and \eqref{def:Hopen} this can be accomplished by starting from a vector in the occupation basis
\[
\left\{ \ket{\mu} : \mu = (\mu_a, \ldots, \mu_b)\in \{0,1\}^{|\Lambda|}\right\} \subseteq \cH_\Lambda
\]
and then taking the subspace spanned by all occupation states arising from successive action of the interaction terms. Since the occupation basis is invariant under the electrostatic interaction terms $n_{x}n_{x+2}$, any subspace which is invariant under all  $q_x^*q_x$ that contribute to $H_\Lambda^\#$, $\# \in\{\cdot, \, \per\}$, will necessarily be an invariant subspace for the respective Hamiltonian. 
A benefit of this approach is that the constructed subspace remains invariant after removing one or more terms from the Hamiltonian. In this case, the original subspace decomposes into a direct sum of invariant subspaces for the new Hamiltonian. This is particularly useful when one considers (1) both the Hamiltonian with periodic and open boundary conditions or (2) how $H_{\Lambda'}\equiv H_{\Lambda'}\otimes \1_{\Lambda\setminus \Lambda'}$ acts on an invariant subspace of $H_\Lambda$ where $\Lambda'\subseteq \Lambda$. Both of these situations will occur in this work.

We focus on identifying invariant subspaces that contain ground states. Due to frustration-freeness, the ground state space for either choice of boundary condition is given by
\[
\caG_\Lambda := \ker(H_\Lambda) = \bigcap_{x=1}^{L-2} \ker n_xn_{x+2}\cap \bigcap_{x=1}^{L-3}\ker q_x, \qquad \caG_\Lambda^\per := \ker(H_\Lambda^\per) = \bigcap_{x=1}^L\left(\ker n_xn_{x+2}\cap \ker q_x\right)
\]
where $\Lambda = [1,L]$. The constraint imposed by the electrostatic terms in the Hamiltonian implies that ground states can only be linear combinations of  $\ket{\mu}$ with $\mu_x\mu_{x+2}= 0$ for all sites $x$ relevant to the boundary conditions. States belonging to $\ker q_x$ can hold at most two particles on $[x,x+3]$. On this interval and at this maximal filling, $q_x$ has exactly three zero eigenstates that satisfy the electrostatic condition, namely
\begin{equation}\label{qx_kernel}
\ket{1100}, \quad \ket{0011}, \quad \text{and}\quad \ket{1001}+\lambda\ket{0110}.
\end{equation}
Moreover, any configuration with one or no particles on $[x,x+3]$ is necessarily in $\ker q_x$. The entire ground state space can be deduced from these simple observations. To organize these in a way that succinctly describes invariant subspaces and ground states, it is convenient to introduce domino tilings of $ \Lambda $. We begin by considering periodic boundary conditions.

\subsection{Periodic boundary conditions}\label{sec:Tiling_Spaces}

 A \emph{(periodic) VMD tiling} $T$ is any covering of the ring $\Lambda = [1,L]$ by void $V$, monomer $M$, and dimer $D$ domino tiles, which are defined as follows:
\begin{enumerate}\setlength{\itemsep}{0pt}

	\item The \emph{void} tile $V=(0)$ covers one lattice site and contains no particles.
	\item The \emph{monomer} tile $M = (100)$ covers three lattice site and has a single particle on the first site.
	\item The \emph{dimer} tile $D=(011000)$ covers six lattice sites and has a particle on the second and third sites. 
\end{enumerate}
The last eigenstate from \eqref{qx_kernel} motivates the reversible replacement rule
\begin{equation}\label{replacement_rule}
	(100)(100) \leftrightarrow (011000)
\end{equation}
which exchanges two neighboring monomers by a dimer (or vice versa). It is clear that by replacing all dimer tiles by a pairs of monomer, each VMD tiling $T$ is transformed to a tiling $R$ consisting of just voids and monomers. Such a tiling $R$ is called a \emph{periodic root tiling}, see Figure~\ref{fig:periodic_tilings}. 
\begin{figure}
	\begin{center}
	\includegraphics[scale=.3]{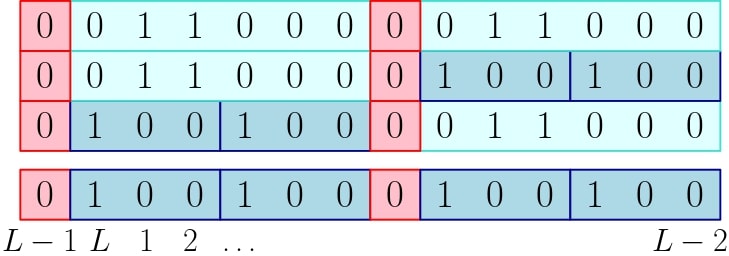}
\end{center}
	\caption{An example of a periodic root tiling $R$ and all connected tilings $T\leftrightarrow R$ for $L=14$. Since the replacement rule only applies to monomers, we use the void at site $x=L-1$ to unravel the ring to an interval for convenience.}
	\label{fig:periodic_tilings}
\end{figure}
More generally, the replacement rule constitutes an equivalence relation on the set of all periodic VMD tilings, denoted by $\cT_\Lambda^\per$. Two tilings $T,T'$ are \emph{connected}, which we denote by $T\leftrightarrow T'$, if one is transformed into the other after a finite number of replacements. The equivalence classes are thus classified by the set of root tilings, $\cR_\Lambda^\per$:
\be\label{tiling_decomp}
\cT_\Lambda^\per = \biguplus_{R\in \cR_\Lambda^\per}\cT_\Lambda^\per(R), \quad \cT_\Lambda^\per(R) = \{ T\in \cT_\Lambda^\per : T\leftrightarrow R\}.
\ee
As shown in the next lemma, the map $\sigma_\Lambda : \cT_\Lambda^\per \to \{0,1\}^L$ that identifies a tiling with its associated configuration is injective.
\begin{lem}\label{lem:configs}
Fix a ring $\Lambda $ of $ L$ sites and let $\sigma_\Lambda:\cT_\Lambda^\per\to \{0,1\}^{L}$ be the mapping that associates a VMD tiling with its configuration . Then $\mu = (\mu_1,\ldots, \mu_L)\in \ran\sigma_\Lambda$ if and only if the follow conditions hold:
\begin{enumerate}
	\item $\mu_{x\pm 2}=0$ for every $x\in\Lambda$ such that $\mu_{x}=1$ and $\mu_{x\pm 1}= 0$.
	\item For every $x\in\Lambda$ such that $\mu_{x}=\mu_{x+1}=1$: 
	\begin{enumerate}
		\item The first three sites on either side of this pair are unoccupied, i.e., $\mu_{s\pm 3/2}=\mu_{s\pm 5/2} = \mu_{s\pm 7/2}= 0$ where $s=x+1/2$.
		\item If the next possible site is occupied, it is proceeded by a vacant site, i.e., if $\mu_{s+ 9/2}= 1$ then $\mu_{s+11/2}=0$, and if $\mu_{s-9/2}=1$ then $\mu_{s-11/2}=0$ where $s=x+1/2$.
	\end{enumerate}
\end{enumerate}
Moreover, the tiling $T \in \cT_\Lambda^\per$ for which $\mu = \sigma_\Lambda(T)$ is unique and so $\sigma_\Lambda$ is injective. (Here, all additions are understood modulo $L$.)
\end{lem}
The elementary proof is a simplified adaption of \cite[Lemma~2.1]{WY:2021} and the details are left to the reader. The forward direction is trivial and only requires counting the number of unoccupied sites between particles on different tiles. The idea for the reverse direction is based off the observation that, for any $\mu\in\ran\sigma_\Lambda$, pairs of nearest neighbor occupied sites must be covered by a dimer, while isolated occupied sites must be covered by a monomer. One then first lays all dimers, then monomers, and finally voids and uses Conditions 1 and 2 to verify at each step that a newly placed tile will not overlap with any previously laid tile. This also constructs the unique tiling $T$ associated with each configuration satisfying Conditions $1$ and $2$ to prove injectivity.\\

Two immediate consequences of the partition of VMD tilings~\eqref{tiling_decomp} and the orthogonality of configuration states are
\be \label{tiling_properties}
\caC_\Lambda^\per = \bigoplus_{R\in\cR_\Lambda^\per} \caC_\Lambda^\per(R) \quad \text{and} \quad \caC_\Lambda^\per(R)\perp \caC_\Lambda^\per(R') \;\; \text{for all} \;\; R\neq R'
\ee
where $C_\Lambda^\per$ is the \emph{subspace of all VMD tilings} and $\caC_\Lambda^\per(R)$ is \emph{VMD subspace associated to $R$}, i.e.
\begin{equation} \label{per_tiling_space}
\caC_\Lambda^\per = \spa\{\ket{\sigma_\Lambda(T)} : T\in\cT_\Lambda^\per\}, \qquad \caC_\Lambda^\per(R) = \spa\{\ket{\sigma_\Lambda(T)} : T\in\cT_\Lambda^\per(R)\}.
\end{equation}
 In particular, these satisfy the following properties.

\begin{lem}\label{lem:invariance} The subspace $\caC_\Lambda^\per(R)$ is invariant under $H_\Lambda^\per$ for each periodic root tiling $R$. Moreover, it supports a unique ground state
\be
\psi_\Lambda^\per(R) = \sum_{T\in\cT_\Lambda^\per(R)} \lambda^{d(T)} \ket{\sigma_\Lambda(T)}
\ee
where $d(T)$ is the number of dimers $D$ used in the tiling $T$.
\end{lem}
\begin{proof}
We compute the action of $q_x$ and $q_x^*q_x$ on an arbitrary tiling state $\ket{\sigma_\Lambda(T)}\in\caC_\Lambda^\per(R)$. The individual tiles are constructed so $\ket{\sigma_\Lambda(T)}$ is in the kernel of each electrostatic interaction term $n_xn_{x+2}$ and
$ q_x\ket{\sigma_\Lambda(T)} = 0 $
unless $T$ has either two consecutive monomers or a dimer beginning at $x$. Thus, to show both the subspace invariance and unique ground state it is sufficient to restrict our attention to these $T$.	
	
Fix $x\in \Lambda$ and assume $T=: T_M\in\cT_\Lambda^\per(R)$ has two monomers starting at $x$. Then, the tiling $T_D$ obtained from replacing these monomers with a dimer also belongs to $\cT_\Lambda^\per(R)$. A direct computation shows
\begin{align*}
	q_x^*q_x \ket{\sigma_\Lambda(T_M)} & = |\lambda|^2\ket{\sigma_\Lambda(T_M)} -\lambda \ket{\sigma_\Lambda(T_D)} \\	
	q_x^*q_x \ket{\sigma_\Lambda(T_D)} & = -\overline{\lambda}\ket{\sigma_\Lambda(T_M)} + \ket{\sigma_\Lambda(T_D)}.
\end{align*}
The invariance property immediately follows since both vectors on the RHS belong to $\caC_\Lambda^\per(R)$. A similar calculation produces 
\be \label{gs_coeff}
q_x\left(c_{T_M} \ket{\sigma_\Lambda(T_M)} + c_{T_D} \ket{\sigma_\Lambda(T_D)}\right) = (c_{T_D}-\lambda c_{T_M}) \ket{\sigma_\Lambda(T_V)}
\ee
for any coefficients $c_{T_M},c_{T_D}\in\bC$, where $T_V$ is the tiling obtained from replacing the two monomers in $T_M$ with six voids. Consider an arbitrary state
\[\psi = \sum_{T\in\cT_\Lambda^\per(R)}c_T\ket{\sigma_{\Lambda}(T)} \in \ker(q_x).\] 
Since $T_V'\neq T_V$ for any other tiling $T_M'\in\cT_\Lambda^\per(R)$ with two consecutive monomers starting at $x$, the injectivity of $\sigma_\Lambda$ implies $c_{T_D}=\lambda c_{T_M}$ by \eqref{gs_coeff}. By frustration-freeness, if $\psi$ is a ground state of $H_\Lambda^\per$, then iterating this procedure over all $x\in\Lambda$ yields $c_T = \lambda^{d(T)}c_R$, from which the claim follows.
\end{proof}

The main result in this subsection shows that there are no other ground states beyond those identifies in the previous lemma for rings of length $ L \geq 6$. A similar result can also be proved for $L=4,5$. However, the definitions of the tiles would need to be adjusted to fit the smaller volume.
\begin{theorem}\label{thm:gss}
	For any $\Lambda = [1,L]$ with $L\geq 6$, the ground state space is a subspace of the tiling space, i.e. $\caG_\Lambda^\per \subseteq \caC_\Lambda^\per$. As a consequence,
	\[
	\{\psi_\Lambda^\per(R) : R\in\cR_\Lambda^\per\}
	\]
	is an orthogonal basis of the ground state space.
\end{theorem}
\begin{proof}
Consider any $\psi=\sum_{\mu\in\{0,1\}^L} \psi(\mu)\ket{\mu} \in \ker(H_\Lambda^\per)$ and rewrite
\be\label{expected_energy}
\braket{\psi}{H_\Lambda^\per \psi} = \sum_{\mu\in\{0,1\}^L}e_\Lambda(\mu)|\psi(\mu)|^2 + \sum_{\nu\in\{0,1\}^L}\sum_{x=1}^L|\braket{\nu}{q_x\psi}|^2, \qquad e_\Lambda(\mu):=\sum_{x=1}^L\mu_x\mu_{x+2} . 
\ee
Introducing the convention that $\psi(\emptyset)=0$, it is straightforward to see that 
\be\label{dipole_energy}
|\braket{\nu}{q_x\psi}|^2 = \left|\psi(\alpha_{x+1}^*\alpha_{x+2}^*\nu)-\lambda\psi(\alpha_{x}^*\alpha_{x+3}^*\nu)\right|^2
\ee
where for each $x\in\Lambda$ the configuration raising operator 
\be\label{raising_op}
\alpha_x^*: \{0,1\}^L\cup \{\emptyset\}\to \{0,1\}^L\cup\{\emptyset\}\ee 
is defined so that if $\mu_x = 0$, then $\mu\mapsto \alpha_x^*\mu$ is the configuration obtained from increasing $\mu_x$ by one, and otherwise $\alpha_x^*\mu = \emptyset$. Similarly, the configuration lowering operator $\alpha_x: \{0,1\}^L\cup \{\emptyset\}\to \{0,1\}^L\cup\{\emptyset\}$ is such that $\mu\mapsto\alpha_x\mu$ decreases $\mu_x$ by one if $\mu_x=1$, and $\alpha_x\mu =\emptyset$ otherwise.

We use the characterization in Lemma~\ref{lem:configs} to show that $ \mu \in \ran\sigma_\Lambda $ if $\psi(\mu)\neq 0$.  

If $e_\Lambda(\mu)> 0$, then $\psi(\mu)=0$ by \eqref{expected_energy}. Thus, assume the electrostatic condition $e_\Lambda(\mu)= 0$ holds. If $\mu$ is the zero configuration, then it trivially satisfies the conditions of Lemma~\ref{lem:configs}. Otherwise, pick $x\in\Lambda$ so that $\mu_x =1$. The electrostatic condition implies $ \mu_{x\pm 2} = 0 $ and $  \mu_{x-1}\mu_{x+1}=0 $. 
	If $\mu_{x\pm 1} = 0$, then Condition 1 of Lemma~\ref{lem:configs} holds. Otherwise we are left to consider whether Condition 2 is satisfied.
	Assume without loss of generality that $\mu_{x+1}=1$ and $\mu_{x-1}=0$. Then $\mu_{x+3}=0$ again by the electrostatic condition. Setting $s=x+1/2$, this collectively implies that $\mu_{s\pm 3/2}=\mu_{s\pm 5/2}=0$. If $\mu_{s+7/2}=\mu_{x+4}=1$, then by \eqref{dipole_energy},
	\[
	0=|\braket{\nu}{q_{x+1}\psi}|^2 = |\psi(\eta)-\lambda\psi(\mu)|^2
	\]
	where $\nu := \alpha_{x+1}\alpha_{x+4}\mu$ and $\eta := \alpha_{x+2}^*\alpha_{x+3}^*\nu$. However, $e_\Lambda(\eta)\geq \eta_x\eta_{x+2}=1$ and so $\psi(\mu)=\psi(\eta)=0$. The analogous argument holds if $\mu_{s-7/2}=1.$ Thus, Condition 2(a) must be satisfied if $\psi(\mu)\neq 0$.
	
	To check Condition 2(b), assume that $\mu_{s+9/2}=\mu_{s+11/2}=1$. Since $s+9/2=x+5$, the electrostatic condition guarantees $\mu_{x+7}=0$, and applying \eqref{dipole_energy},
	\[
		0=|\braket{\nu}{q_{x+4}\psi}|^2 = |\psi(\mu)-\lambda\psi(\eta)|^2
	\]
	where $\nu := \alpha_{x+5}\alpha_{x+6}\mu$ and $\eta := \alpha_{x+4}^*\alpha_{x+7}^*\nu$. However, $\psi(\eta)=0$ by the previous case, and so once again $\psi(\mu)=0$. An analogous argument holds if $\mu_{s-9/2}=\mu_{s-11/2}=1$. 
	
	Thus, one concludes that $\psi(\mu)\neq 0$ only if $\mu \in\ran\sigma_\Lambda$, completing the proof.
\end{proof}

\subsection{Open boundary conditions}\label{sec:BVMD}

A similar (but slightly more complicated) tiling construction describes the invariant subspaces that support ground states of the Hamiltonian with open boundary conditions. The relevant lattice tilings come from restricting periodic VMD tilings to a subinterval. Additional boundary tiles emerge when the restriction cuts through the interior of a monomer or dimer tile. Listing only those truncated tiles that cannot be built from other tiles, this produces the following set of boundary tiles:
\begin{enumerate}
	\item \emph{On the left boundary:}
	\begin{enumerate}
		\item The dimer $B_l = (11000)$ covering five sites with particles on the first two sites.
	\end{enumerate}
	\item \emph{On the right boundary:}
	\begin{enumerate}
		\item The monomer $M_1 = (1)$ covering one site with a single particle.
		\item The monomer $M_2=(10)$ covering two sites with a particle on the first site.
		\item The dimer $D_1 = (0110)$ covering four sites with a particle on the second and third sites.
		\item The dimer $D_2 = (01100)$ covering five sites with a particle on the second and third sites.
		\item The dimer $B_r = (011)$ covering three sites with particles on the last two sites.		
	\end{enumerate}
\end{enumerate}
A \emph{BVMD tiling} $T$ of the interval $\Lambda = [a,b]$ is any covering of $\Lambda$ consisting of voids $V$, monomers $M$, dimers $D$, and boundary tiles. Of course, one can place any of the bulk tiles $\{V, \, M, \, D\}$ at the boundary. However, boundary tiles may only be placed on their respective boundary.

The truncated monomers, $M_i$, and dimers, $D_i$, for $i=1,2$ give rise to new (reversible) replacement rules:
\be \label{truncated_replacements}
(100)(10)\leftrightarrow (01100), \qquad (100)(1)\leftrightarrow (0110) \, .
\ee
The dimers $B_l$ and $B_r$ cannot be replaced as the Hamiltonian preserves particle number and center of mass. Combining these rules with \eqref{replacement_rule} once again produces an equivalence relation on the set of all BVMD tilings $\cT_\Lambda$, where we say two tilings are connected, $T'\leftrightarrow T$, if one tiling can be transformed into the other with a finite number of replacements.  The equivalence classes are again labeled by $\cR_\Lambda = \{R\in\cT_\Lambda : R \text{ a root tiling}\}$ analogous to \eqref{tiling_decomp} where we classify $R$ as a \emph{root tiling} if it does not contain any dimers $D_i$, $i=1,2,3$, where $D_3:=D$, see Figure~\ref{fig:obc_tilings}.

\begin{figure}
	\begin{center}
		\includegraphics[scale=.48]{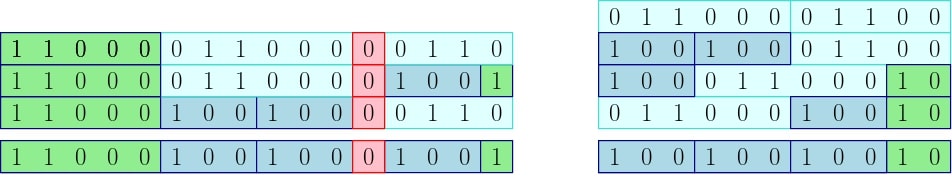}
	\end{center}
\caption{The set of all connected tilings $T\leftrightarrow R$ associated with a root tiling $R$ on an interval of 16 and 11 sites, respectively. The second is the set of $T\leftrightarrow M_4^{(2)}$ associated with the squeezed Tao-Thouless state $\vp_4^{(2)}$.}
\label{fig:obc_tilings}
\end{figure}

Severing a periodic VMD tiling between the first and last site of the ring creates a BVMD tiling. As a result, the set of VMD tilings $\cT_\Lambda^\per$ can be identified with a subset of $\cT_\Lambda$. The configuration map $\sigma_\Lambda: \cT_\Lambda \to \{0,1\}^{|\Lambda|}$ remains injective after extending the domain, and the classification from Lemma~\ref{lem:configs} again holds with the minor adjustment that the conditions are trivially assumed to be satisfied for any $x\in\Lambda^c$ rather than considering addition modulo $L$. Analogous versions of the other properties from Section~\ref{sec:Tiling_Spaces} hold for the BVMD tilings. To state these, we denote by
\be\caC_\Lambda = \spa\{\ket{\sigma_\Lambda(T)} : T\in \cT_\Lambda\}, \qquad \caC_\Lambda(R) = \{\ket{\sigma_\Lambda(T)} : T\leftrightarrow R\}
\ee
the \emph{space of all BVMD tilings} and the \emph{BVMD subspace associated with $R\in\cR_\Lambda$}, respectively. 

\begin{theorem}\label{thm:BVMD}
Fix an interval $\Lambda = [a,b]$. The following properties apply to each root tiling $R\in \cR_\Lambda$:
\begin{enumerate}
	\item $\caC_\Lambda(R) \perp \caC_\Lambda(R')$ for any root $R'\neq R$.
	\item $\caC_\Lambda(R)$ is invariant under $H_\Lambda$. As a consequence, so is $\caC_\Lambda = \bigoplus_{R\in\cR_\Lambda}\caC_\Lambda(R)$.
	\item $\cG_\Lambda\cap \caC_\Lambda(R) = \spa\{\psi_\Lambda(R)\}$ where
	\be\label{BVMD_state}
	\psi_\Lambda(R) = \sum_{T\leftrightarrow R} \lambda^{d(T)}\ket{\sigma_\Lambda(T)}
	\ee
	and $d(T)$ is the number of dimers $D_i$, $i=1,2,3,$ in the tiling $T$. 
	\item If $|\Lambda|\geq 8$, then $\ker(H_\Lambda) = \spa\{\psi_\Lambda(R) : R\in\cR_\Lambda\}$.
\end{enumerate}
\end{theorem}

The proof is a slight variation of arguments in Section~\ref{sec:Tiling_Spaces}. An alternate proof of the last property using a different approach can be found in \cite[Theorem 2.15]{NWY:2021}.

Several other properties of the BVMD states $\psi_\Lambda(R)$ will be important for the analysis in Section~\ref{sec:MM}. Any tiling $T\in\cT_\Lambda$ can be written in terms of its ordered tiling  $T=(T_1,T_2,\ldots T_k)$ where 
\[
T_1 \in\{B_l\}\cup \cD_{\textrm{bulk}}, \quad T_k\in\{B_r, M_i, D_i : i = 1,2\}\cup \cD_{\textrm{bulk}}. 
\]
and $T_i \in \cD_{\textrm{bulk}}:= \{V,M,D\}$ for all $2\leq i \leq k-1$. Setting $M_3=M$, a special case is the root tiling
\[
M_L^{(i)} = (M,M,\ldots, M, M_i)
\]
which covers an interval $\Lambda_i$ of $3(L-1)+i$ sites with $L$ monomers for $i=1,2,3$, see Figure~\ref{fig:obc_tilings}. The corresponding ground state $\varphi_L^{(i)}:=\psi_{\Lambda_i}(M_L^{(i)})$ is called a \emph{squeezed Tao-Thouless state}, and all BVMD states $\psi_\Lambda(R)$ can be written (up to boundary terms) as a product of such states and void states $\ket{0}$. Namely, if $\{v_1, \ldots, v_n\}$ denotes the ordered set of sites covered with voids by $R=(R_1,\ldots, R_k)$, then
\be\label{fragmentation}
\psi_{\Lambda}(R) = \psi^l\otimes\vp_{L_0}\otimes\ket{0}_{v_1}\otimes\ldots\otimes\vp_{L_{n-1}}\otimes\ket{0}_{v_n}\otimes\psi^r
\ee
where $\vp_L :=\vp_L^{(3)}$, $L_i$ is the number of monomers between $v_i$ and $v_{i+1}$ in $R$, $L_0$ is the number of monomers to the left of $v_1$, $L_{n}$ is the number of monomers to the right of $v_n$ and
\[
\psi^l = \begin{cases} \ket{11000}  & R_1 = B_l \\ 1 & \text{otherwise}  \end{cases},
\qquad
\psi^r = \begin{cases} \vp_{L_{n}}\otimes\ket{011}  & R_k = B_r \\ \vp_{L_{n}}^{(i)}  & R_{k}=M_i \\ 1 & \text{otherwise}  \end{cases},
\]
see Figure~\ref{fig:obc_tilings}. The choice of the right boundary condition, and placement of the last void $v_n$ will play a critical role in classifying the tiling spaces $\caC_\Lambda(R)$ and, hence, the BVMD states $\psi_\Lambda(R)$ for the martingale method in Section~\ref{sec:MM}.
 
The squeezed Tao-Thouless states $\vp_L^{(i)}$ also satisfy a number of useful recursion relations. To state them, we slightly abuse the notation and write
\[
\ket{D_1} = \ket{0110}, \quad \ket{D_2} = \ket{01100}, \quad \ket{D_3} = \ket{011000}
\]
as well as the simplified form $\ket{D}:=\ket{D_3}$ for the configuration states associated with the dimers $D_i$. Similarly we write $\ket{M_i}$ for the configuration state associated with a monomer of length $i=1,2,3$.
Then for each $i\in\{1,2,3\}$ the squeezed Tao-Thouless states satisfy the following properties:
\begin{enumerate}
	\item For any $1\leq j<i$ and $n\in\bN$, $\vp_n^{(i)}=\vp_n^{(j)}\otimes \ket{0}^{\otimes i-j}$.
	\item For any $l,r\in\bN$ with $r\geq 2$
	\be\label{recursion1}
	\vp_{l+r}^{(i)} = \vp_l\otimes\vp_r^{(i)} + \lambda \vp_{l-1}\otimes \ket{D}\otimes \vp_{r-1}^{(i)}.
	\ee
	\item For any $n\geq 2$
	\be\label{recursion2}
	\vp_n^{(i)} = \vp_{n-1}\otimes \ket{M_i} + \lambda \vp_{n-2}\otimes \ket{D_i}.
	\ee
	As a consequence, $
	\|\vp_n^{(i)}\|^2 = \|\vp_{n-1}\|^2+|\lambda|^2\|\vp_{n-2}\|^2$ and the ratio 
	\be\label{beta_n}
	\beta_n:=\frac{\|\vp_{n-1}\|^2}{\|\vp_n\|^2} =\frac{1}{\beta_+}\frac{1-\beta^n}{1-\beta^{n+1}}\ee
	is a convergent sequence where  $\beta_{\pm} = (1\pm \sqrt{1+4|\lambda|^2})/2$ and $\beta = \frac{\beta_-}{\beta_+}\in(-1,0)$.
\end{enumerate}
The various state decompositions are immediate after considering the set of tilings $T\leftrightarrow M_n^{(i)}$ and the definition of $\psi_{\Lambda_i}(M_n^{(i)})$, see Figure~\ref{fig:obc_tilings} and \eqref{BVMD_state}. The norm equality is trivial as the two states on the RHS of \eqref{recursion2} are orthogonal. A full proof of~\eqref{beta_n} is given in \cite[Lemma 2.13]{NWY:2021}.

\subsection{Embeddings and isospectrality}\label{sec:isospectral}

Applying a spectral-gap method to a local Hamiltonian on some Hilbert space $\cH_\Lambda$ requires considering the action of local Hamiltonians $H_{\Lambda'}$ associated to smaller volumes $\Lambda'\subseteq\Lambda$. For the model at hand, the tiling spaces will play the role of the local Hilbert space.  To illuminate a useful isospectral relationship, we examine the relation between the (B)VMD tiling spaces associated with the Hamiltonians
\be\label{Ham_relations}
H_{\Lambda'} \leq H_\Lambda \leq H_\Lambda^\per
\ee
where $\Lambda' \subseteq \Lambda = [1,L]$ and we use the standard identification $H_{\Lambda'}:=H_{\Lambda'}\otimes\1_{\Lambda \setminus \Lambda'}$.

Both tiling spaces $\caC_\Lambda^\#$, with $\#\in\{\cdot , \per\}$, from Sections~\ref{sec:Tiling_Spaces}-\ref{sec:BVMD} were defined so that they are invariant under all interaction terms associated $H_\Lambda^\#$. Thus, these subspaces are then also invariant under all of the interaction terms associated with $H_{\Lambda'}$, and so
\be\label{invariance_embedding}
H_{\Lambda'}\caC_\Lambda^\# \subseteq \caC_\Lambda^\# .
\ee
The Hamiltonian $H_{\Lambda'}$ only acts non-trivially on the portion of a tiling state $\ket{\sigma_\Lambda(T)}$ that covers $\Lambda'$. This corresponds to the configuration obtained from truncating the tiling $T$ to $\Lambda'$. The boundary tiles were introduced in Section~\ref{sec:BVMD} so that the restriction agrees with a tiling on $\Lambda'$. Hence, there is a unique $T\restriction_{\Lambda'}\in\cT_{\Lambda'}$ such that
\be\label{restriction}
\sigma_\Lambda(T)\restriction_{\Lambda'} = \sigma_{\Lambda'}(T\restriction_{\Lambda'}).
\ee
A natural question is whether $\caC_\Lambda^\#$ can be written in terms of BVMD tiling spaces $\caC_{\Lambda'}(R')$.

\begin{lem}\label{lem:decomp}
	Fix $\Lambda'=[a,b]\subseteq [1,L]=\Lambda$. Then
	\be\label{direct_sum2}
	\caC_\Lambda = \bigoplus_{R'\in\cR_{\Lambda'}}\bigoplus_{\substack{\mu\in\ran\sigma_\Lambda \, :\\ \mu = (\mu^l,\sigma_{\Lambda'}(R'),\mu^r)}} \ket{\mu^l}\otimes \caC_{\Lambda'}(R')\otimes\ket{\mu^r}.
	\ee
Moreover, if $|\Lambda|\geq |\Lambda'|+4$, the same equality holds if  one replaces $\caC_\Lambda$ in the LHS with $\caC_\Lambda^\per$, and $\sigma_\Lambda$ in the RHS with $\sigma_\Lambda\restriction_{\cT_\Lambda^\per}$.
\end{lem}

In the above, $\mu_l$ and $\mu_r$ are the subconfigurations of $\mu$ supported on the subinterval of $\Lambda$ to the left and right of $\Lambda'$, respectively. In the case that one of these subintervals is empty, we use the convention $\ket{\mu^\#} = 1$. Here, we also use the notation
\[
\psi_l\otimes \cV \otimes \psi_r := \{\psi_l\otimes\psi\otimes\psi_r: \psi\in\cV\}
\]
for a subspace $\cV\subseteq \cH$. The constraint $|\Lambda|\geq |\Lambda'|+4$ for the case of periodic boundary conditions here and Corollary~\ref{cor:isospectral} below simply guarantees that all BVMD tilings on $\Lambda'$ can extend to a periodic VMD tiling on $\Lambda$.

The proof follows from the same reasoning used in~\cite[Lemma~3.3]{WY:2021}. Since the truncated replacement rules from \eqref{truncated_replacements} agree with the original rule \eqref{replacement_rule} after appending one or two additional zeros, the claimed equality \eqref{direct_sum2} is a consequence of \eqref{restriction} noticing that for any $T'\leftrightarrow R'\in\cR_{\Lambda'}$:
\be \label{replacement_equivalence}
(\mu^l,\sigma_{\Lambda'}(T'),\mu^r)\in\ran\sigma_\Lambda \iff (\mu^l,\sigma_{\Lambda'}(R'),\mu^r)\in\ran\sigma_\Lambda.
\ee
This can be checked using the open boundary conditions version of Lemma~\ref{lem:configs}. We illustrate \eqref{replacement_equivalence} in Figure~\ref{fig:truncated_tilings}, and leave the details of its proof to the reader.\\

\begin{figure}
	\begin{center}
		\includegraphics[scale=.21]{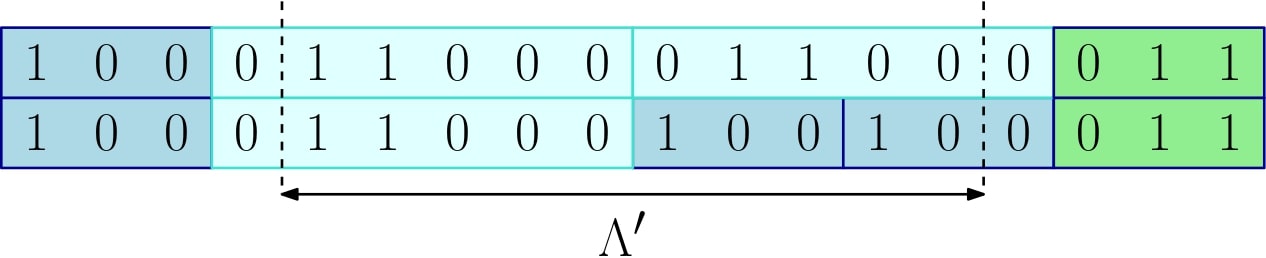}
	\end{center}
\caption{A restricted tiling $T'=T\restriction_{\Lambda'}\in\cT_{\Lambda'}$. Freezing the configurations outside of $\Lambda'$ and replacing $T'$ with any $T''\leftrightarrow T'$ again produces a tiling configuration on $\Lambda$.}
\label{fig:truncated_tilings}
\end{figure}

%

The block decomposition in Lemma~\ref{lem:decomp} implies that any operator $A\in\cB(H_{\Lambda'})$ that leaves $\caC_{\Lambda'}(R')$ invariant for all $R'\in\cR_{\Lambda'}$ also leaves $\caC_\Lambda^\#$ invariant for each $\#\in\{\cdot ,\per\}$. This guarantees the following isospectral relationships which are used in the application of the martingale method and finite size criterion in Sections~\ref{sec:MM}-\ref{sec:FSC} below.

\begin{cor}\label{cor:isospectral}
Let $A\in\cB(\cH_{\Lambda'})$ be an operator which leaves $\caC_{\Lambda'}(R')$ invariant for each $R'\in\cR_{\Lambda'}$. Then:
\begin{enumerate}
	\item $\|A\otimes\1_{\Lambda\setminus\Lambda'}\|_{\caC_\Lambda} = \|A\|_{\caC_{\Lambda'}}$, where the subscript denotes the Hilbert space with which the norm is taken.
	\item If $A^*=A$, then $\spec(A\otimes\1_{\Lambda\setminus\Lambda'}\restriction_{\caC_\Lambda})=\spec(A\restriction_{\caC_{\Lambda'}})$
\end{enumerate}
Moreover, if $|\Lambda|\geq |\Lambda'|+4$, the same relations hold when one replaces $\caC_\Lambda$ with $\caC_\Lambda^\per$.
\end{cor}

\begin{proof}
	Since $A$ leaves $\caC_{\Lambda'}(R')$ invariant for each $R'\in\cR_{\Lambda'}$, Lemma~\ref{lem:decomp} guarantees $\caC_{\Lambda}^\#$ is invariant under $A\otimes\1_{\Lambda\setminus\Lambda'}$, and so
	\[
	A\otimes\1_{\Lambda\setminus\Lambda'}\restriction_{\caC_\Lambda^\#} = P_{\caC_\Lambda^\#}(A\otimes\1_{\Lambda\setminus\Lambda'})P_{\caC_\Lambda^\#}
	\]
	where $P_{\caC_\Lambda^\#}$ is the orthogonal projection onto $\caC_{\Lambda^\#}$. Moreover, as the BVMD-subspaces $\{\caC_{\Lambda'}(R'): R'\in\cR_{\Lambda'}\}$ form family of orthogonal subspaces, $A\otimes\1_{\Lambda\setminus\Lambda'}$ is block diagonal with respect to the decomposition from Lemma~\ref{lem:decomp}. Recalling that $\caC_{\Lambda'}=\bigoplus_{R'}\caC_{\Lambda'}(R')$ and each $\caC_{\Lambda'}(R')$ is represented on the RHS of \eqref{direct_sum2}, the second property is immediate. Additionally recognizing that
	\[
	\|A\otimes\1_{\Lambda\setminus\Lambda'}\|_{\caC_\Lambda^\#} = \|P_{\caC_\Lambda^\#}(A\otimes\1_{\Lambda\setminus\Lambda'})P_{\caC_\Lambda^\#}\|_{\cH_\Lambda}, \qquad \|A\|_{\caC_{\Lambda'}} = \|P_{\caC_{\Lambda'}}AP_{\caC_{\Lambda'}}\|_{\cH_\Lambda}
	\]
	the first property again follows from the block diagonalization \eqref{direct_sum2}.
 	
%
\end{proof}

We end this section by providing an orthogonal basis for the intersection 
\[
\caC_\Lambda\cap(\caG_{\Lambda'}\otimes \cH_{\Lambda\setminus\Lambda'}).
\]
The situation of interest for the gap estimate produced in Section~\ref{sec:MM} is when $\Lambda'$ and $\Lambda$ only differ by the last three sites, i.e. 
\[\Lambda':=[1,L-3]\subseteq [1,L] = :\Lambda.\] 
Using frustration-freeness, one immediately has
\be\label{ff_gs_cond}
\psi_\Lambda(R) \in \caG_{\Lambda'}\otimes \cH_{\Lambda \setminus\Lambda'}
\ee
 for all $ R\in\cR_\Lambda $, but this set is incomplete. Knowing how a single $\caC_\Lambda(R)$ decomposes in terms of the BVMD spaces $\caC_{\Lambda'}(R')$ will help identify how to extend this to a basis. For $\Lambda'\subseteq\Lambda$ as above, this decomposition is completely characterized by whether the replacement rules apply to the last two tiles of an arbitrary root tiling $R=(R_1, \ldots, R_k)\in\cR_\Lambda$, and so we set
\be\label{RMM}
\cR_\Lambda^{MM}= \{R\in\cR_\Lambda \, | \, R \text{ ends in two or more monomers}\}.
\ee

For any $R\in \cR_{\Lambda}\setminus\cR_{\Lambda}^{MM}$, the particle content of the last three sites of any $T\leftrightarrow R$ is left invariant as the replacement rules do not apply. As a consequence, for all such $R$,
\be\label{intersection1}
\caC_\Lambda(R) = \caC_{\Lambda'}(R')\otimes \ket{\sigma_{\Lambda}(R)\restriction_{\Lambda \setminus\Lambda'}}, 
\ee
where $R' = R\restriction_{\Lambda'}$ is such that \eqref{restriction} holds, see Figure~\ref{fig:truncations_MM}. 
\begin{figure}
	\begin{center}
		\includegraphics[scale=.35]{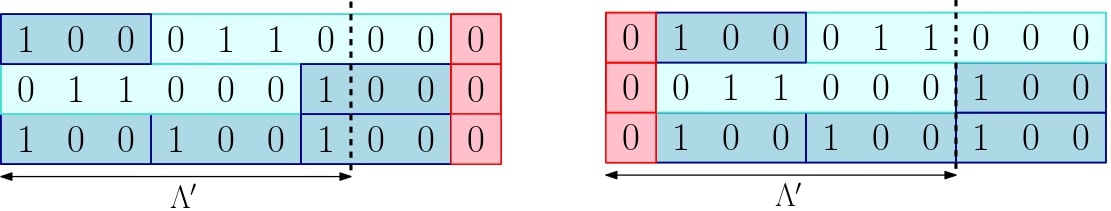}
	\end{center}
\caption{An example of the decomposition of $\caC_\Lambda(R)$ for a root $R\in\cR_\Lambda\setminus\cR_\Lambda^{MM}$ and $R\in\cR_\Lambda^{MM}$, respectively.}
\label{fig:truncations_MM}
\end{figure}

On the other hand, if $R=(R_1, \ldots, R_{k-2}, M, M_{i})\in\cR_{\Lambda}^{MM}$ for some $i\in\{1,2,3\}$, the set of tiles $T\leftrightarrow R$ can be partitioned into two sets: those ending in $M_{i}$, and those ending in $D_i$. Setting $R_D = (R_1, \ldots, R_{k-2},D_i)$, this produces
\be\label{intersection2}
\caC_\Lambda(R) = \left(\caC_{\Lambda'}(R')\otimes \ket{\sigma_{\Lambda}(R)\restriction_{\Lambda \setminus\Lambda'}}\right)\oplus \left(\caC_{\Lambda'}(R_D')\otimes \ket{\sigma_{\Lambda}(R_D)\restriction_{\Lambda \setminus\Lambda'}} \right)
\ee
where $R'=R\restriction_{\Lambda'}$ and $R_D' = R_D\restriction_{\Lambda'}$, see Figure~\ref{fig:truncations_MM}. 

 Recalling that each BVMD space $\caC_{\Lambda'}(R')$ supports a unique ground state by Theorem~\ref{thm:BVMD}, it is immediately clear from \eqref{intersection1}-\eqref{intersection2} that
\be\label{gs_dim}
\dim \left(\caC_\Lambda(R)\cap(\caG_{\Lambda'}\otimes \cH_{\Lambda \setminus\Lambda'})\right) = \begin{cases}
	1 & R\in \cR_{\Lambda}\setminus\cR_{\Lambda}^{MM} \\ 2 & R\in \cR_{\Lambda}^{MM}\\
\end{cases} \, .
\ee
Hence, given the direct sum decomposition of $\caC_\Lambda$ and the orthogonality properties from Theorem~\ref{thm:BVMD}, to extend the BVMD states $\{\psi_\Lambda(R) : R\in\cR_\Lambda\}$ to an orthogonal basis of $\caC_\Lambda\cap \caG_{\Lambda'}\otimes \cH_{\Lambda\setminus\Lambda'}$, one only needs to find a single vector
\be\label{gs_requirement}
\eta_\Lambda(R)\in\caC_{\Lambda}(R)\cap(\caG_{\Lambda'}\otimes \cH_{\Lambda \setminus\Lambda'}) \quad \text{such that}\quad \braket{\eta_\Lambda(R)}{\psi_\Lambda(R)} =0
\ee
for each $R\in \cR_{\Lambda}^{MM} $. For any such root tiling there is a unique $n\geq 2$ and $i\in \{1,2,3\}$ so that the ordered tiling of $R$ can be partitioned into two sub-roots $R=(\tilde{R},M_n^{(i)})$ where $\tilde{R}$ does not end in a monomer. (In the case that $R=M_n^{(i)}$ we use the convention that $\tilde{R} = \emptyset$.) The next result shows that
\be\label{eta}
\eta_\Lambda(R):=\psi_{\Lambda(n,i)}(\tilde{R})\otimes\eta_n^{(i)}
\ee
satisfies both requirements from \eqref{gs_requirement} where 
\be\label{eta_n}
\eta_n^{(i)}:= -\overline{\lambda}\alpha_{n-1}\vp_{n-1}\otimes\ket{M_i}+\vp_{n-2}\otimes\ket{D_i}
\in\caC_{[1,3(n-1)+i]}(M_n^{(i)})
\ee
and $\Lambda(n,i)\subseteq\Lambda$ is the subinterval covered by $\tilde{R}$.

\begin{lem}\label{lem:orth_basis}
	Suppose that $\Lambda'=[1,L-3]\subseteq [1,L]=\Lambda$ for some $L\geq 4$. Then
	\be
	\{\psi_\Lambda(R) : R\in \cR_\Lambda\} \cup \{ \eta_\Lambda(R) : R\in\cR_{\Lambda}^{MM}\}
	\ee
	is an orthogonal basis for $\caC_\Lambda\cap \caG_{\Lambda'}\otimes \cH_{\Lambda \setminus\Lambda'}$.
\end{lem}

\begin{proof}
Given \eqref{gs_dim} and since 
\[\caC_\Lambda\cap \caG_{\Lambda'}\otimes \cH_{\Lambda \setminus\Lambda'} = \bigoplus_{R\in\cR_\Lambda}\left(\caC_\Lambda(R)\cap \caG_{\Lambda'}\otimes \cH_{\Lambda \setminus\Lambda'}\right),\]
we only need to verify that $\eta_\Lambda(R)$ defined as in \eqref{eta} satisfies \eqref{gs_requirement}.  
Recalling the recursion relation from \eqref{recursion2}, it is straightforward to calculate $\braket{\vp_n^{(i)}}{\eta_n^{(i)}}=0$ for all $n\geq 2$ and $i\in\{1,2,3\}$. By the fragmentation property \eqref{fragmentation} and \eqref{eta},
\[\braket{\eta_\Lambda(R)}{\psi_\Lambda(R)} = \|\psi_{\Lambda(n,i)}(\tilde{R})\|^2\braket{\vp_n^{(i)}}{\eta_n^{(i)}}, \]  
which implies the orthogonality condition~\eqref{gs_requirement}. 

Using the three properties of the squeezed Tao-Thouless state surrounding \eqref{recursion1}-\eqref{beta_n}, the first term on the RHS of \eqref{eta_n} can be rewritten in order to conclude that there are roots $R_1, R_2\in\cR_{\Lambda'}$ so that
\[
\psi_{\Lambda'}(R_1)=\psi_{\Lambda(n,i)}(\tilde{R})\otimes\vp_{n-1}^{(i)}, \quad \psi_{\Lambda'}(R_2)=\psi_{\Lambda(n,i)}(\tilde{R})\otimes\vp_{n-2}\otimes \ket{\tilde{D}_i}
\]
where $\ket{\tilde{D}_1}=\ket{0}, \, \ket{\tilde{D}_2}=\ket{01},$ and $\ket{\tilde{D}_3}=\ket{011}$. This proves that 
$\eta_\Lambda(R)\in \cG_{\Lambda'}\otimes\cH_{\Lambda \setminus\Lambda'}$.
\end{proof}

\section{Uniform bulk gap via invariant subspaces}\label{sec:proof}

In this section, we adapt the novel method from \cite{WY:2021} of establishing bulk spectral gaps in the presence of edge states via invariant subspaces to the present case.  Let us briefly review the general approach.

As we saw in Section~\ref{sec:Tiling_Spaces}, the periodic Hamiltonian $H_\Lambda^\per$ is block diagonal with respect to decomposition $\cH_\Lambda = \caC_\Lambda^\per \oplus (\caC_{\Lambda}^\per)^\perp$. The gap above its ground state space $\caG_\Lambda^\per \subseteq  \caC_\Lambda^{\per} $  therefore is 
\[
\gap (H_\Lambda^\per) = \min\{ E_1(\caC_{\Lambda}^\per), \, E_0((\caC_{\Lambda}^\per)^\perp)\}
\]
where
\[
 E_1(\caC_{\Lambda}^\per) = \inf_{\substack{\psi \in \caC_\Lambda^\per \cap (\caG_\Lambda^\per)^\perp \\ \psi\neq 0}}\frac{\braket{\psi}{H_\Lambda^\per \psi}}{\|\psi\|^2}, \qquad 
  E_0((\caC_{\Lambda}^\per)^\perp) = \inf_{0\neq\xi \in (\caC_{\Lambda}^\per)^\perp }\frac{\braket{\xi}{H_\Lambda^\per \xi}}{\|\xi\|^2}.
\]
For our proof of Theorem~\ref{thm:bulk_gap}., we establish separate bounds on $E_1(\caC_{\Lambda}^\per)$ and $E_0((\caC_{\Lambda}^\per)^\perp)$ that are (1) uniform in the volume and (2) robustly positive in the limit $\lambda \to 0$. The estimate on $E_0((\caC_{\Lambda}^\per)^\perp)$ in Section~\ref{sec:electrostatics} is model specific and  utilizes the characterization of tiling-state configurations in Lemma~\ref{lem:configs} to produce electrostatic estimates. The bound on $E_1(\caC_{\Lambda}^\per)$ produced in Section~\ref{sec:FSC} relies on a version of Knabe's finite size criteria from \cite{knabe:1988}. In turn, this estimate depends on a uniform lower bound on
\be\label{gap_obc}
\gap(H_{\Lambda'} \restriction_{\caC_\Lambda^\per}) = \inf_{0\neq\psi \in \caC_\Lambda^\per\cap \caG_{\Lambda'}^\perp}\frac{\braket{\psi}{H_{\Lambda'}\psi}}{\|\psi\|^2}
\ee
for $\Lambda'\subseteq \Lambda$ sufficiently large, which is proved in Section~\ref{sec:MM} using the martingale method \cite{nachtergaele:1996, nachtergaele:2016b}. Since~\eqref{gap_obc} involves the Hamiltonian $ H_{\Lambda'} $ with open boundary, one could worry that edge states, whose energies tend to zero as $ \lambda \to 0 $  and which hinder the estimates of the bulk gap in \cite{NWY:2021},  could again destroy the required robust estimates on $E_1(\caC_\Lambda^\per)$. The reason this does not occur is that all of the low-lying edge states of $H_{\Lambda'}$ belong to $(\caC_{\Lambda}^\per)^\perp$ and hence do not enter~\eqref{gap_obc}. Cutting off the edge states before applying martingale or finite-volume analysis is the core idea of the new method from~ \cite{WY:2021}.

\subsection{Martingale method}\label{sec:MM}
As a consequence of  the isospectral and embedding properties in Corollary~\ref{cor:isospectral}, a lower bound on the gap in~\eqref{gap_obc} is immediately obtained from a lower bound on
\be\label{eq:OBC_uniform_gap}
E_1(\caC_\Lambda):= \sup_{0\neq\psi\in  \caC_\Lambda\cap \cG_{\Lambda}^\perp}\frac{\braket{\psi}{H_\Lambda\psi}}{\|\psi\|^2}
\ee
for any sufficiently large interval $\Lambda = [1,L]$. This will be established via the martingale method from~\cite{nachtergaele:2016b}. The idea behind this approach is to use a sequence of operators to effectively trap a single excitation to a finite region independent of $\Lambda$. For the truncated $\nu=1/3$ model, the method will approximately localize an excitation to an interval with 8 to 10 sites, which is captured by the norm bound \eqref{confined} below.

To set up the method, write $L=3N+k$ for some $N\geq 2$ and $k\in\{1,2,3\}$, and define two finite sequences of Hamiltonians $h_n,H_n\in \cB(\cH_\Lambda)$
\be\label{h_n}
h_n = H_{\Lambda_n}, \quad H_n =\sum_{m=2}^n h_m, \quad \Lambda_n = \begin{cases} [1,6+k] & n = 2 \\ [3n+k-8,3n+k], & 2<n\leq N
	\end{cases}
\ee
for $2 \leq n \leq N$. Furthermore, denote by $g_n$ and $G_n$ the orthogonal projections onto the corresponding ground state spaces $\ker(h_n)=\cG_{\Lambda_n}\otimes\cH_{\Lambda\setminus\Lambda_n}$ and $\ker(H_n)$, respectively. Every interaction term ($n_{x}n_{x+2}$ or $q_x^*q_x$) is supported on at least one and at most three of the intervals $\Lambda_n$. As each interaction is non-negative, for all $  2\leq n \leq N $, 
\be\label{equiv_Hams}
H_{[1,3n+k]} \leq H_{n} \leq 3H_{[1,3n+k]} , 
\ee
and, in particular, the ground state spaces agree, i.e. $\ker(H_n)=\cG_{[1,3n+k]}\otimes\cH_{\Lambda\setminus [1,3n+k]}$. Finally, define a resolution of the identity defined in terms the ground state projections:
\begin{equation}
	\label{E_n}
	E_n := \begin{cases}
		\1-G_2, & n=1 \\
		G_n-G_{n+1}, & 2\leq n \leq N-1 \\
		G_N, & n=N .
	\end{cases}
\end{equation}

We consider the restrictions of these operators to $\caC_\Lambda$. As discussed in Section~\ref{sec:isospectral},  both $H_n$ and $h_n$ (and, thus, their corresponding ground-state projections) leave $\caC_\Lambda$ invariant. To simplify notation, denote by $\cA^\cV$ the restriction of an operator $A\in\cB(\cH_\Lambda)$ to an invariant subspace $\cV\subseteq \cH_\Lambda$, i.e.
\be\label{invariance}
A^\cV : = A\restriction_{\cV} = P_{\cV} A P_{\cV}
\ee 
where $P_{\cV}$ is the orthogonal projection onto $\cV$. In particular, invariance implies $A=A^{\cV} + A^{\cV^\perp}$.

\begin{theorem}\label{thm:MM}
	Fix $\Lambda = [1,L]$ with $L\geq 10$. The restrictions of the operators $h_n$, $H_n$, $g_n$ and $E_n$ to $\caC_\Lambda$, defined as in \eqref{h_n}-\eqref{invariance}, satisfy the following three properties for all $2\leq n \leq N$:
	\begin{enumerate}
		\item $h_n^{\caC_\Lambda}\geq \kappa(\1-g_n)^{\caC_\Lambda}$
		\item $[g_n^{\caC_{\Lambda}}, E_m^{\caC_{\Lambda}}] \neq 0$ only if $m\in[n-3,n-1]$.
		\item For$|\lambda|\neq 0$, the ground state projections satisfies
		\be\label{confined}
		\|g_{n}^{\caC_{\Lambda}}E_{n-1}^{\caC_{\Lambda}}\|^2 \leq f(|\lambda|^2):=\sup_{k\geq 4} f_k(|\lambda|^2)
		\ee
		where, given $\beta_k$ from \eqref{beta_n},
		\begin{equation}\label{f_k}
			f_k(r) = r\beta_k\beta_{k-2}\left(
			\frac{[1-\beta_{k-1}(1+r)]^2}{1+2r}
			+\beta_{k-3}\frac{r(1-\beta_{k-1})^2}{1+r}
			\right).
		\end{equation}
	\end{enumerate} 
As a consequence, if $|\lambda|>0$ and $f(|\lambda|^2)<1/3$, the spectral gap of $H_\Lambda\restriction\caC_\Lambda$ is bounded from below by
\be
E_1(\caC_\Lambda) \geq \frac{\kappa}{3}\left(1-\sqrt{3f(|\lambda|^2)}\right)^2.
\ee
\end{theorem}

The main challenge to prove of Theorem~\ref{thm:MM}, which is left to the end of this subsection, is to establish the third property. 

\begin{lem}\label{lem:epsilon}
	Fix $|\lambda|>0$ and $L\geq 10$, and define $\Lambda = [1,L]$, $\Lambda_1 = [1,L-3]$ and $\Lambda_2 = [L-8,L]$. Then
	\be\label{epsilon_calc}
	\|G_{\Lambda_2}(\1-G_\Lambda)G_{\Lambda_1}\|_{\caC_{\Lambda}}^2 \leq f(|\lambda|^2)
	\ee
	where $G_{\Lambda'}$ is the orthogonal projection onto the ground state space $\caG_{\Lambda'}\otimes\cH_{\Lambda\setminus\Lambda'}$.
\end{lem}

\begin{proof}
The frustration-free property and Theorem~\ref{thm:BVMD} guarantee that $\caG_{\Lambda}\subseteq \cD_{\Lambda_1}^{\Lambda}:=\caG_{\Lambda_1}\otimes\cH_{\Lambda\setminus\Lambda_1}\cap\caC_{\Lambda}$. As a consequence, the norm can be expressed as
\be\label{rewrite}
\|G_{\Lambda_2}(\1-G_\Lambda)G_{\Lambda_1}\|_{\caC_{\Lambda}}^2 = \sup_{0\neq \psi\in \caG_\Lambda^\perp \cap \cD_{\Lambda_1}^{\Lambda}}\frac{\|G_{\Lambda_2}\psi\|^2}{\|\psi\|^2}.
\ee
The subspace $\cD_{\Lambda_1}^{\Lambda}$ is of the form considered in Lemma~\ref{lem:orth_basis} and since the BVMD states $\psi_\Lambda(R)$ form a basis for the ground-state space $\cG_\Lambda$, this yields
\[
\caG_\Lambda^\perp\cap\cD_{\Lambda_1}^{\Lambda} = \spa\{\eta_\Lambda(R) : R\in\cR_\Lambda^{MM}\} .
\]
 We first bound $\|G_{\Lambda_2}\eta_\Lambda(R)\|^2$ for each $R\in\cR_{\Lambda}^{MM}$ and then generalize this to an arbitrary $\psi\in\caG_\Lambda^\perp\cap\cD_{\Lambda_1}^{\Lambda}$. Factoring $\eta_\Lambda(R)$ as in \eqref{eta}, the first step breaks into two cases determined by whether or not $\Gamma(n,i):=\Lambda\setminus \Lambda(n,i)$ is a subset of $\Lambda_2$.

If $n=2,3$, then $\Gamma(n,i)\subseteq \Lambda_2$ as $|\Gamma(n,i)|\leq 9$, and the frustration-free property of the ground-state space implies that $G_{\Lambda_2}=G_{\Lambda_2}G_{\Gamma(n,i)}$. Therefore, applying \eqref{eta},
\be\label{small_n}
G_{\Lambda_2}\eta_\Lambda(R) = G_{\Lambda_2}\left(\psi_{\Lambda(n,i)}(\tilde{R})\otimes G_{\Gamma(n,i)}\eta_n^{(i)}\right)=0 , 
\ee
where the last equality holds since $\eta_n^{(i)},\vp_n^{(i)}\in\caC_{\Gamma(n,i)}(M_n^{(i)})$ are orthogonal, and $\eta_n^{(i)}$ is orthogonal to all other BVMD states in $\caC_{\Gamma(n,i)}$ by the subspace orthogonality relations, see Theorem~\ref{thm:BVMD}.

\begin{figure}
	\begin{center}
		\includegraphics[scale=.125]{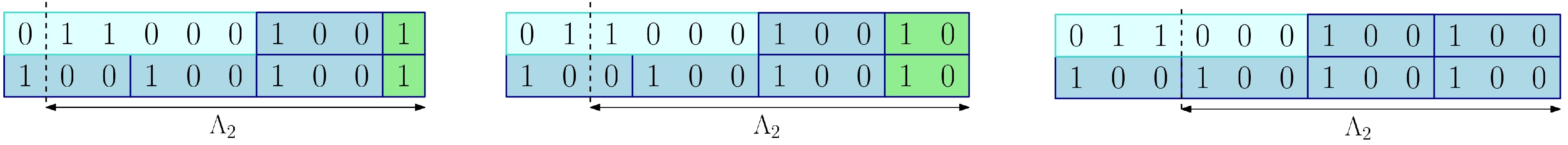}
	\end{center}
\caption{The root tilings $R_1^i$ and $R_2^i$ for $i=1,2,3$ obtained by truncating tilings $T\in\caC_{\Gamma(n,i)}(M_N^{(i)})$ to $\Lambda_2$.}\label{fig:MM_truncations}
\end{figure}

If $n\geq 4$, then applying Theorem~\ref{thm:BVMD}, $G_{\Lambda_2}$ can be expressed (up to a factor of the identity) as
\be
G_{\Lambda_2} = \sum_{R\in\cR_{\Lambda_2}}\frac{\ketbra{\psi_{\Lambda_2}(R)}}{\|\psi_{\Lambda_2}(R)\|^{2}},
\ee
and $G_{\Lambda_2}\eta_\Lambda(R) = \psi_{\Lambda(n,i)}(\tilde{R})\otimes G_{\Lambda_2}\eta_n^{(i)}$ as the support of the projection satisfies $\Lambda_2\subseteq \Gamma(n,i)$.
We only need to consider the states $\psi_{\Lambda_2}(R)$ that have a nonzero overlap with $\eta_n^{(i)}\in\caC_{\Gamma(n,i)}(M_n^{(i)})$. These are identified by the set of tilings produced from restricting $M_n^{(i)}\restriction_{\Lambda_2}$. For each case $i\in\{1,2,3\}$, there are precisely two roots $R_1^i,R_2^i\in\cR_{\Lambda_2}$ (sketched in Figure~\ref{fig:MM_truncations}) whose connected tilings $T\leftrightarrow R_k^i$ are obtained as such restrictions. Using \eqref{eta_n} and the recursion relations~\eqref{recursion1}--\eqref{recursion2} to rewrite $\eta_n^{(i)}$ and these two BVMD states, a direct calculation then yields
\begin{align}
G_{\Lambda_2}\eta_n^{(i)} & =\sum_{m=1,2}\frac{\ketbra{\psi_{\Lambda_2}(R_m^i)}}{\|\psi_{\Lambda_2}(R_m^i)\|^{2}}\eta_n^{(i)}\nonumber\\
& = \frac{\overline{\lambda}(1-\beta_{n-1}\|\vp_2\|^2)}{\|\vp_3\|^2}\vp_{n-3}\otimes\vp_{3}^{(i)} + \frac{|\lambda|^2(1-\beta_{n-1})}{\|\vp_2\|^2}\vp_{n-4}\otimes\ket{D}\otimes\vp_2^{(i)}. \label{G2_action}
\end{align}
The norm of $\eta_n^{(i)}$ can be expressed as $\|\eta_n^{(i)}\|^2 = \|\vp_{n-3}\|^2/(\beta_n\beta_{n-2})$ using \eqref{eta_n}. Combined with \eqref{G2_action} and substituting $\|\vp_3^{(i)}\|^2 = 1+2|\lambda|^2$ and $\|\vp_2^{(i)}\|=1+|\lambda|^2$, this produces the final estimate
\be\label{G_eta}
\|G_{\Lambda_2}\eta_\Lambda(R)\|^2 = f_n(|\lambda|^2)\|\eta_\Lambda(R)\|^2 . 
\ee

The mutual orthogonality of the BVMD spaces from Theorem~\ref{thm:BVMD} implies that $\{G_{\Lambda_2}\eta_\Lambda(R) : R\in\cR_\Lambda^{MM}\}$ is again a set of orthogonal states as each $\caC_{\Lambda}(R)$ is invariant under $G_{\Lambda_2}$. Thus, for an arbitrary $\psi \in \cD_{\Lambda_1}^\Lambda$, 
\[
\|G_{\Lambda_2}\psi\|^2 \leq \sup_{n\geq 4}f_{n}(|\lambda|^2)\|\psi\|^2 = f(|\lambda|^2)\|\psi\|^2,
\]
by \eqref{small_n} and \eqref{G_eta}. Hence, \eqref{epsilon_calc} follows from \eqref{rewrite}.
\end{proof}

We can now prove the lower bound on $E_1(\caC_\Lambda)$ from Theorem~\ref{thm:MM}.

\begin{proof}[Proof of Theorem~\ref{thm:MM}]
	If Properties 1-3 hold, then by the martingale method \cite[Theorem 5.1]{nachtergaele:2016b}, \[\gap(H_N^{\caC_{\Lambda}}):=\sup_{0\neq \psi\in\cG_\Lambda^\perp\cap\caC_\Lambda}\frac{\braket{\psi}{H_N\psi}}{\|\psi\|^2}\geq \kappa(1-\sqrt{3f(|\lambda|^2)})^2.\]
	The operator inequality \eqref{equiv_Hams} still holds when the Hamiltonians are restricted to the invariant subspace $\caC_\Lambda$. It implies  $E_1(\caC_{\Lambda})\geq \gap(H_N^{\caC_{\Lambda}})/3,$ and the claim follows. Thus, one needs only verify the three properties.
	
	1. By translation invariance, one trivially has for all $n\geq 2$,
	\[
	h_n^{\caC_{\Lambda}} \geq \min_{k\in\{7,8,9\}}E_1(\caC_{[1,k]})(\1-g_n)^{\caC_{\Lambda}}
	\]
	as $h_n = H_{\Lambda_n}$ with $|\Lambda_n|\in \{7,8,9\}$. The value $\kappa = \inf_{k\in\{7,8,9\}}E_1(\caC_{[1,k]})$ is obtained by calculating the gap in each BVMD space $\caC_{[1,k]}(R)$ and taking the minimum. This establishes Property 1.
	
	2. Since $\caC_\Lambda$ is invariant under each $g_n$ and $E_m$, the commutator can be expressed as
	\[
	[g_n^{\caC_{\Lambda}}, E_m^{\caC_{\Lambda}}] = P_{\caC_\Lambda}[g_n,E_m]P_{\caC_\Lambda}
	\]
	where $P_{\caC_\Lambda}$ is the orthogonal projection onto $\caC_{\Lambda}$. For $m<n-3$, the support of $E_m$ and $g_n$ is disjoint and the operators commute. For $m>n-1$, the support of $g_n=G_{\Lambda_n}$ is contained in the support of $G_m = G_{[1,3m+k]}$ and $G_{m+1} = G_{[1,3(m+1)+k]}$. Hence, $g_nE_m=E_mg_n =E_m$ for such $m$ by \eqref{E_n} and frustration-freeness. This establishes Property 2.
	
	3. The claim is trivial for $n=2$ as $g_2E_1=0$. For any $3\leq n \leq N$, similar to Property 2, the norm can be rewritten as
	\[
	\|g_{n}^{\caC_{\Lambda}}E_{n-1}^{\caC_{\Lambda}}\| = \|P_{\caC_{\Lambda}}g_{n}E_{n-1}P_{\caC_{\Lambda}}\|= \|g_{n}E_{n-1}\|_{\caC_{\Lambda}} .
	\]
	Notice that $\supp(g_{n}E_{n-1}) \subseteq \tilde{\Lambda}_n:=[1,3n+k]$ and, moreover, $\caC_{\tilde{\Lambda}_n}(R)$ is invariant under $g_{n}E_{n-1}$ for all $R\in\cR_{\tilde{\Lambda}_n}$. Thus, by Corollary~\ref{cor:isospectral},
	\[
	\|g_{n}E_{n-1}\|_{\caC_{\Lambda}} = \|g_{n}E_{n-1}\|_{\caC_{\tilde{\Lambda}_n}} = \|G_{\tilde{\Lambda}_n\setminus \tilde{\Lambda}_{n-3}}(\1-G_{\tilde{\Lambda}_n})G_{\tilde{\Lambda}_{n-1}}\|_{\caC_{\tilde{\Lambda}_n}}
	\]
	where we have inserted the definitions of $g_n$ and $E_{n-1}$, and applied the frustration-free property to factor $E_{n-1}$. The last expression is of the form treated in Lemma~\ref{lem:epsilon}, which establishes Property 3.
\end{proof}

\subsection{Finite size criterion}\label{sec:FSC}

It is now possible to produce a bound on $E_1(\caC_\Lambda^\per)$ that is both $|\Lambda|$ independent and is strictly positive in the limit $\lambda \to 0$. To do so, we employ the variant of Knabe's finite size criterion \cite{knabe:1988} found in \cite[Theorem 3.10]{NWY:2021} together with the open-boundary-conditions bound from Theorem~\ref{thm:MM}. The final bound will additionally depend on two constants
\be\label{spectral_bounds}
\kappa =\min_{k\in\{6,7,8\}}E_1(\caC_{[1,k]}),
\qquad \kappa(1+2|\lambda|^2) =\max_{k\in\{6,7,8\}}\|H_{[1,k]}\|_{\caC_{[1,k]}} 
\ee
which are calculated by determining the gap, resp.\ norm, of the Hamiltonian restricted to each BVMD space $\caC_{[1,k]}(R)$ and then taking the minimum, resp. maximum, over all $R$ and $k$.

\begin{theorem}\label{thm:FCS} Fix $n\geq 2$. Then for any ring $\Lambda = [1,L]$ with $L\geq 3n+9$,
	\be
	E_1(\caC_{\Lambda}^\per) \geq\frac{n}{2(1+2|\lambda|^2)(n-1)}\left[
	\min_{3\leq k \leq 5}E_1(\caC_{[1,3n+k]})-\frac{\kappa(1+2|\lambda|^2)}{n}\right].
	\ee
\end{theorem}

The proof proceeds in a series of general operator inequalities that imply gap bounds. Namely, the kernels of a pair of non-negative operators $A,B\in\cB(\cH)$ necessarily agree if there constants $c,C> 0$ such that $cB \leq A \leq CB$. If this kernel is nontrivial, then it is a ground-state space of both operators, and the inequality lifts to their respective spectral gaps, i.e.
\be\label{gap_equivalence}
c\gap(B) \leq \gap(A) \leq C\gap(B).
\ee
For simplicity, we will refer to such a pair as \emph{gap equivalent operators}. For a frustration-free Hamiltonian $H$, a simple way to construct a gap equivalent operator is to take $P=\1-G$ where $G$ is the orthogonal projection onto its ground state space, $\ker(H)$. This is sometimes referred to as the \emph{spectrally flattened Hamiltonian}, and it satisfies 
\be \label{spec_flattened}
\gap(H) P \leq H \leq \|H\|P.
\ee
Spectrally flattened Hamiltonians can be used to construct course-grained versions of a frustration-free quantum spin model which are amenable to finite size criterion. This is the approach taken here. 

\begin{proof}
We write $L=3N+r$ for some $r\in\{3,4,5\}$. Similar to the martingale method, we introduce a sequence of intervals $\Lambda_m$, $1\leq m \leq N+1$, that cover the ring $\Lambda$ in such a way that every interaction term ($n_xn_{x+2}$ or $q_x^*q_x$) is supported on at least one and at most two of the intervals; specifically,
\be
\Lambda_m = 
\begin{cases}
	[3m-2,3m+3], & 1\leq m \leq N, \\
	[L-r+1,L+3] & m=N+1,
\end{cases}
\ee
where we identify $x\equiv x+L$ in the ring geometry. Thus, the operator bounds
\be\label{Ham_bounds}
H_{\Lambda_{n,k}} \leq \sum_{m=k}^{n+k-1}H_{\Lambda_m} \leq 2 H_{\Lambda_{n,k}}, \quad H_{\Lambda}^\per\leq \sum_{m=1}^{N+1}H_{\Lambda_m} \leq 2 H_\Lambda^\per
\ee
hold for each $1\leq k \leq N+1$, where $\Lambda_{n,k}=\bigcup_{m=k}^{n+k-1}\Lambda_m$ and the addition $n+k-1$ is understood modulo $N+1$. These inequalities also hold when the operators are restricted to $\caC_\Lambda^\per$ as this is an invariant subspace of all of the Hamiltonians. Notice that the intervals $\Lambda_{n,k}$ satisfy $3n+3\leq |\Lambda_{n,k}|\leq 3n+5$ for all $k$. Since $|\Lambda|\geq 3n+9$, translation invariance and Corollary~\ref{cor:isospectral} imply the first gap bound
\be\label{gap:step1}
\min_{1\leq k \leq N+1}\gap(H_{\Lambda_{n,k}}\restriction_{\caC_\Lambda^\per}) \geq \min_{3\leq k\leq 5}E_1(\caC_{[1,3n+k]}).
\ee

We replace $H_{\Lambda_{n,k}}\restriction_{\caC_\Lambda^\per}$ and $H_\Lambda^\per\restriction_{\caC_\Lambda^\per}$ with gap equivalent Hamiltonians that are amenable to finite-size criteria. These are obtained from the spectrally flattened Hamiltonians $P_m :\caC_\Lambda^\per \to \caC_{\Lambda}^\per$ that are the orthogonal projections onto $\ran(H_{\Lambda_m}\restriction_{\caC_\Lambda^\per})$. Corollary~\ref{cor:isospectral} implies that $\spec(H_{\Lambda_m}\restriction_{\caC_\Lambda^\per})=\spec(H_{\Lambda_m}\restriction_{\caC_{\Lambda_m}})$. Therefore, by \eqref{spec_flattened}
\[
\kappa P_m \leq H_{\Lambda_m}\restriction_{\caC_\Lambda^\per} \leq \kappa(1+2|\lambda|^2) P_m
\] 
where we invoke \eqref{spectral_bounds} since $|\Lambda_m|\in\{6,7,8\}$. Summing the above over appropriate values of $m$ and applying the restricted form of \eqref{Ham_bounds} produces
\be\label{gap_equiv_hams}
\frac{\kappa}{2}H_{n,k} \leq H_{\Lambda_{n,k}}\restriction_{\caC_\Lambda^\per} \leq \kappa(1+2|\lambda|^2)H_{n,k}, \qquad 
\frac{\kappa}{2}H_{N} \leq H_{\Lambda}^\per\restriction_{\caC_\Lambda^\per} \leq \kappa(1+2|\lambda|^2)H_{N}
\ee
where $H_{n,k},H_N\in\cB(\caC_\Lambda^\per)$ are the gap equivalent Hamiltonians defined by
\[
H_{n,k}=\sum_{m=k}^{n+k-1}P_m, \qquad H_N = \sum_{m=1}^{N+1}P_m.
\]

The second set of operator inequalities in \eqref{gap_equiv_hams} guarantees that $H_N$ is a frustration-free Hamiltonian as the kernel is nontrivial and $P_m\geq 0$ for all $m$. In addition, any pair of distinct intervals $\Lambda_l$ and $\Lambda_m$ are disjoint unless $|l-m|=1$ or $\{l,m\}=\{1,N+1\}$. Since $P_m\psi = (\1-G_{\Lambda_m})\psi$, for all $\psi\in\caC_\Lambda^\per$, this implies that $[P_m,P_l]=0$ under the same constraints. Thus, the Hamiltonians $H_{n,k},H_N\in\cB(\caC_\Lambda^\per)$ satisfy the conditions of \cite[Theorem 3.10]{NWY:2021}, and so the respective spectral gaps satisfy
\be\label{gap:step2}
\gap(H_N) \geq \frac{n-1}{n}\left(\min_{1\leq k \leq N+1}\gap(H_{n,k})-\frac{1}{n}\right).
\ee
Equations \eqref{gap_equivalence} and \eqref{gap_equiv_hams} imply that
\be\label{gap:step3}
E_1(\caC_\Lambda^\per) = \gap(H_{\Lambda}^\per\restriction_{\caC_\Lambda^\per} ) \geq \frac{\kappa}{2}\gap(H_N), \qquad \gap(H_{n,k}) \geq \frac{1}{2\kappa(1+2|\lambda|^2)}\gap( H_{\Lambda_{n,k}}\restriction_{\caC_\Lambda^\per} ) .
\ee
Hence, the result is an immediate consequence of combining the bounds in \eqref{gap:step1} and \eqref{gap:step2}-\eqref{gap:step3}.
\end{proof}

\subsection{Electrostatic estimates}\label{sec:electrostatics}

The desired lower bound on $E_0((\caC_\Lambda^\per)^\perp)$ is the focus of this section. 
The subspace $(\caC_\Lambda^\per)^\perp$ is spanned by the set of configuration states labeled by  the set
\[\cS_\Lambda := \{0,1\}^{|\Lambda|}\setminus\ran(\sigma_\Lambda\restriction_{\cT_\Lambda^\per}).\] 
Our main objective is to identify a constant $\gamma=\gamma(\kappa,\lambda)>0$, which is  independent of $\Lambda$ and strictly positive in the limit $\lambda\to0$, so that the expected energy for any state  $\psi=\sum_{\mu\in\cS_\Lambda}\psi(\mu)\ket{\mu}$ satisfies
\be\label{gs_energy}
\braket{\psi}{H_\Lambda^\per \psi} = \sum_{\mu\in\cS_\Lambda^{(1)} }e_\Lambda(\mu)|\psi(\mu)|^2 + \kappa \sum_{\nu\in\{0,1\}^{|\Lambda|}}\sum_{x\in\Lambda}|\braket{\nu}{q_x\psi}|^2 \geq \gamma\|\psi\|^2,
\ee
where $e_\Lambda(\mu)=\sum_{x\in\Lambda}\mu_x\mu_{x+2}$ is the electrostatic energy associated with $\mu$, and
\[
\cS_\Lambda^{(1)} := \left\{\mu\in\{0,1\}^{|\Lambda|} : e_\Lambda(\mu)\geq 1\right\}\subseteq \cS_\Lambda.
\]

As is evident from \eqref{gs_energy}, the main challenge comes from the configurations $\mu\in\cS_\Lambda\setminus \cS_\Lambda^{(1)}$. Lemma~\ref{lem:configs} implies these can be partitioned into the following two sets (which are understood with the convention $x\equiv x+|\Lambda|$):
\begin{align*}
	\cS_\Lambda^{(2)} & = \left\{\mu\in\{0,1\}^{|\Lambda|} : \mu_x=\mu_{x+1}=1, \, \mu_{x-3}+\mu_{x+4}\geq 1 \text{ for some }x\in\Lambda\right\} \setminus \cS_\Lambda^{(1)}\\
	\cS_\Lambda^{(3)} & = \left\{\mu\in\{0,1\}^{|\Lambda|} : \mu_x=\mu_{x+1}=\mu_{x-4}=\mu_{x-5}=1 \text{ for some }x\in\Lambda\right\} \setminus (\cS_\Lambda^{(1)}\cup\cS_\Lambda^{(2)}).
\end{align*}

 Recalling the notation from \eqref{dipole_energy}, the strategy for the next result is to use this classification to pick a configuration $\nu\in\{0,1\}^{|\Lambda|}$ and site $x\in\Lambda$ for each $\mu\in\cS_\Lambda^{(2)}\cup \cS_\Lambda^{(3)}$ so that $\mu \in \{\alpha_{x+1}^*\alpha_{x+2}^*\nu, \, \alpha_{x}^*\alpha_{x+3}^*\nu\}$, and then estimate 
 \[
 |\braket{\nu}{q_x \psi}|^2 = |\psi(\alpha_{x+1}^*\alpha_{x+2}^*\nu)-\lambda\psi(\alpha_{x}^*\alpha_{x+3}^*\nu)|^2
 \]
 from below by a linear combination of $|\psi(\alpha_{x+1}^*\alpha_{x+2}^*\nu)|^2$ and $|\psi(\alpha_{x}^*\alpha_{x+3}^*\nu)|^2.$

\begin{theorem}\label{thm:electrostatic}
	Suppose that $\Lambda = [1,L]$ with $L \geq 11$ and $|\lambda|>0$. Then, $E_0(\caC_\Lambda^\per)^\perp) \geq \gamma^\per$, where as defined in~\eqref{def:gammaper},
	\[
	\gamma^\per=\frac{1}{3}\min\left\{1,\frac{\kappa}{2+2\kappa|\lambda|^2},\frac{\kappa}{1+\kappa}\right\}.
	\]
\end{theorem}

\begin{proof}
We fix $\psi\in(\caC_\Lambda^\per)^\perp$and consider separately all $\mu\in\cS_\Lambda$ in the support of $ \psi $ to establish \eqref{gs_energy}. 
\vskip 4pt

\emph{ Case $\mu\in\cS_\Lambda^{(1)}$:}  One trivially has the lower bound
\be\label{bound1}
e_\Lambda(\mu)|\psi(\mu)|^2  \geq |\psi(\mu)|^2 =: \gamma^{(1)}|\psi(\mu)|^2.
\ee
For future purpose, we additionally set $ T_\psi(\mu) := 0 $. 

\emph{ Case $\mu\in\cS_\Lambda^{(2)}$:} Set $x\equiv x_\mu = \max\{y\in[1,L] : \mu_{y+1}=\mu_{y+2}=1 \, \wedge \, \mu_{y-2}+\mu_{y+5}\geq 1\}$. Except those sites indicated by this set, all other sites between $[x-2,x+5]$ are unoccupied since $\mu\notin \cS_\Lambda^{(1)}$. Thus, both $\nu = \alpha_{x+1}\alpha_{x+2}\mu$ and $\eta(\mu)\equiv\eta = \alpha_{x}^*\alpha_{x+3}^*\nu$ are nonempty configurations. By the Cauchy-Schwarz inequality, the lower bound
\[T_\psi(\mu) := |\braket{\nu}{q_x\psi}|^2 = |\psi(\mu)-\lambda\psi(\eta)|^2 
\geq (1-\delta)|\psi(\mu)|^2 - |\lambda|^2\frac{1-\delta}{\delta}|\psi(\eta)|^2
\]
holds for all $\delta\in(0,1)$. Choosing $\delta = \frac{\kappa|\lambda|^2}{1+\kappa|\lambda|^2}$, and noting that $e_\Lambda(\eta)\geq 1$ produces the final estimate 
\be\label{bound2}
e_\Lambda(\eta)|\psi(\eta)|^2 + \kappa T_\psi(\mu) \geq \frac{\kappa}{1+\kappa|\lambda|^2}|\psi(\mu)|^2 =: \gamma^{(2)}|\psi(\mu)|^2.
\ee

\emph{ Case $\mu\in\cS_\Lambda^{(3)}$:} Set $x\equiv x_\mu = \max\{y\in[1,L] : \mu_{y+1}=\mu_{y+2}=\mu_{y-3}=\mu_{y-4} = 1\}$. Again, since $\mu\notin \cS_\Lambda^{(1)}\cup\cS_\Lambda^{(2)}$ all other (non-required) sites between $[x-4,x+3]$ are unoccupied. Thus, the following four configurations are nonempty: 
\[
\nu = \alpha_{x+1}\alpha_{x+2}\mu, 
\quad \eta' = \alpha_{x}^*\alpha_{x+3}^*\nu, 
\quad \nu'=\alpha_{x-3}\alpha_{x}\eta', \quad
\eta(\mu)\equiv\eta=\alpha_{x-2}^*\alpha_{x-1}^*\nu',
\] 
and for all $\delta,\delta'\in(0,1)$ the Cauchy-Schwarz inequality yields
\begin{align*}
T_\psi(\mu) :=& |\braket{\nu}{q_x\psi}|^2 + |\braket{\nu'}{q_{x-3}\psi}|^2 
= |\psi(\mu)-\lambda\psi(\eta')|^2 +  |\psi(\eta)-\lambda\psi(\eta')|^2 \\
\geq& (1-\delta)|\psi(\mu)|^2 + |\lambda|^2\left(1-\delta'-\frac{1-\delta}{\delta}\right)|\psi(\eta')|^2 - \frac{1-\delta'}{\delta'}|\psi(\eta)|^2.
\end{align*}
By considering its occupation on $[x-4,x+3]$, it is easy to check that $e_\Lambda(\eta)=1$. Hence, choosing $\delta' = \frac{\kappa}{\kappa+1}$ and $\delta = \frac{1}{2-\delta'}$ gives
\be\label{bound3}
e_\Lambda(\eta)|\psi(\eta)|^2 + \kappa T_\psi(\mu) \geq \frac{\kappa}{1+\kappa}|\psi(\mu)|^2 =: \gamma^{(3)}|\psi(\mu)|^2.
\ee

In each of the three cases $i\in\{1,2,3\}$ the pairs $(x,\nu)$ that contribute to $T_\psi(\mu)$ for any $\mu\in\cS_\Lambda^{(i)}$ are unique. Thus, for fixed $i$, 
\begin{align}
	\gamma^{(i)}\sum_{\mu\in\cS_\Lambda^{(i)}}|\psi(\mu)|^2 & \leq \sum_{\mu\in\cS_\Lambda^{(i)}}\left(e_\Lambda(\eta(\mu))|\psi(\eta(\mu))|^2+\kappa T_\psi(\mu)\right) \nonumber\\
	& \leq c_i\sum_{\eta\in\cS_\Lambda^{(1)}}e_\Lambda(\eta)|\psi(\eta)|^2+\kappa\sum_{\nu\in\{0,1\}^{|\Lambda|}}\sum_{x\in\Lambda}|\braket{\eta}{q_x\psi}|^2 \leq c_i\braket{\psi}{H_\Lambda \psi}
\end{align}
where $c_i := \max_{\eta\in\cS_{\Lambda}^{(1)}}|\{\mu\in\cS_\Lambda^{(i)}: \eta(\mu)=\eta\}|\geq 1$. The claimed bound is then a consequence of dividing by $c_i$ and summing over $i$. Thus, the result follows from determining $c_i$ for $i=1,2,3$.

It is trivial that $c_1 = 1$. The values $c_2=2$ and $c_3=1$ can be determined for noting that there is an interval $\Lambda_x'\subseteq \Lambda$ of at most 8 sites near $x=x_\mu$ that contains all sites of $\eta=\eta(\mu)$ that contribute to the electrostatic energy. The constraint $|\Lambda|\geq 11$ guarantees this interval can be uniquely identified in the ring geometry. In the case of $c_3$, $e_\Lambda(\eta)= \eta_{x-4}\eta_{x-2}$, from which it is possible to identify $x$ and map back to the unique $\mu$. In the case of $c_2$, depending on the value of $\eta_{x-2}+\eta_{x+5}\geq 1$, we have
$ e_\Lambda(\eta) = \eta_{x-2}\eta_{x} + \eta_{x+3}\eta_{x+5}\in \{1,2\} $.
When $e_\Lambda(\eta)=1$, it is not always possible to determine if the electrostatic energy comes from the interval $[x-2,x]$ or $[x+3,x+5]$, which accounts for the value $c_2=2$. 
\end{proof}

A similar approach can be used produce a lower bound
\[
E_0(\caC_\Lambda^\perp) :=\sup_{0\neq \psi \in\caC_\Lambda^\perp}\frac{\braket{\psi}{H_\Lambda\psi}}{\|\psi\|^2} \geq \gamma^\obc
\]
with $0 < \gamma^\obc = \mathcal{O}(|\lambda|^2)$ reflecting the presence of the edge modes discussed in the introduction and \cite{NWY:2021}. The difference in the bound is due to configurations $\mu\in\cS_\Lambda^{(2)}$ where the pair of nearest-neighbor occupied sites are along the boundary of $\Lambda$. 

\subsection{Proof of Theorem~\ref{thm:bulk_gap}}
We now provide the final details in the proof of the main result.
\begin{proof}[Proof of Theorem~\ref{thm:bulk_gap}]
Fix $|\Lambda|\geq 18$ and let $n=n(|\Lambda|)\geq 3$ be the largest integer so that $|\Lambda|\geq 3n+9$. The subspace $\caC_\Lambda^\per$ is invariant under $H_\Lambda^{\per}$. Thus, the Hamiltonian is block diagonal with respect to the Hilbert space decomposition $\cH_\Lambda = \caC_\Lambda^\per \oplus (\caC_\Lambda^\per)^\perp$. Since $\caG_\Lambda^\per \subseteq \caC_\Lambda^\per$, Theorem~\ref{thm:FCS} and Theorem~\ref{thm:electrostatic} imply
\begin{align*}
\gap(H_\Lambda^\per) &=\min\{E_0((\caC_\Lambda^\per)^\perp),\; E_1(\caC_\Lambda^\per)\} \\
& \geq \min \left\{\gamma^\per, \;\frac{1}{2(1+2|\lambda|^2)}\left[
\min_{3\leq k \leq 5}E_1(\caC_{[1,3n+k]})-\frac{\kappa(1+2|\lambda|^2)}{n}\right]\right\}\\
&\geq
\min \left\{\gamma^\per, \;\frac{\kappa}{6(1+2|\lambda|^2)}\left(1-\sqrt{3f(|\lambda|^2)}\right)^2-\frac{\kappa}{2n}\right\}
\end{align*}
where in the last inequality we have used Theorem~\ref{thm:MM} to bound $E_1(\caC_{[1,3n+k]})$. The result immediately follows.
\end{proof}
\minisec{Acknowledgements }

{\small This work was supported by the DFG under EXC-2111--390814868.}

\bibliographystyle{abbrv}  
\bibliography{bosonic} 

\bigskip
\bigskip
\bigskip
\bigskip

\noindent Simone Warzel\\
Munich Center for Quantum Science and Technology, and\\
Zentrum Mathematik  \& Department of Physics, TU M\"{u}nchen\\
85747 Garching, Germany\\
\verb+warzel@ma.tum.de+\\

\noindent Amanda Young\\
Munich Center for Quantum Science and Technology, and\\
Zentrum Mathematik, TU M\"{u}nchen\\
85747 Garching, Germany\\
\verb+young@ma.tum.de+\\

\end{document}